\documentclass[11pt]{article}
\pdfoutput=1
\usepackage[text={6.5in,9in}]{geometry}

\usepackage[utf8]{inputenc}
\usepackage{mathptmx}
\usepackage{subfig}

\usepackage[dvips]{color}
\usepackage{graphicx}
\usepackage{paralist,tabularx}
\usepackage{amsmath,amsfonts,amssymb,amsthm}
\usepackage{latexsym}
\usepackage{enumerate}
\usepackage{mathtools}
\usepackage{wrapfig}

\usepackage{algorithm, algorithmicx, algpseudocode} 
\usepackage{algorithmicx}

\usepackage{algorithm}
\usepackage{algorithmicx}
\usepackage{algpseudocode}
\usepackage[algo2e]{algorithm2e}

\newenvironment{reminder}[1]{\smallskip
	\noindent {\bf Reminder of #1 }\em}{}

\newtheorem{theorem}{Theorem}[section]
\newtheorem{lemma}[theorem]{Lemma}

\newtheorem{corollary}[theorem]{Corollary}
\newenvironment{definition}[1][Definition]{\begin{trivlist}
		\item[\hskip \labelsep {\bfseries #1}]}{\end{trivlist}}
	
{\makeatletter
 \gdef\xxxmark{%
   \expandafter\ifx\csname @mpargs\endcsname\relax 
     \expandafter\ifx\csname @captype\endcsname\relax 
     \marginpar{xxx}
          \else
       xxx 
    \fi
   \else
     xxx 
   \fi}
 \gdef\xxx{\@ifnextchar[\xxx@lab\xxx@nolab}
 \long\gdef\xxx@lab[#1]#2{\textbf{[\xxxmark #2 ---{\sc #1}]}}
 \long\gdef\xxx@nolab#1{\textbf{[\xxxmark #1]}}
}

\newcommand{\kstr}{k^\star}

\newcommand{\RandomSat}{\textsc{$k$-RandomSat }}
\newcommand{\HammDist}{\textsc{$\alpha$-SampleAndTest}}

\newcommand{\ncs}{\textsc{NumClausesSAT}}
\newcommand{\ncus}{\textsc{NumClausesUnSAT}}

\newcommand{\smlHDAlg}{\textsc{SAT-from-$\alpha$-Small-HD}}
\newcommand{\pprom}{p_{\textrm{promising}}}
\newcommand{\pbad}{p_{\textrm{bad}}}
\newcommand{\Sbad}{S_{\textrm{bad}}}
\newcommand{\DPa}{D_{pa}}
\newcommand{\DPc}{D_{pc}}

\newcommand{\eps}{\epsilon}

\newcommand{\HD}{HammingDistance}
\newcommand{\kval}{60}

\begin{document}
\title{Faster Random $k$-CNF Satisfiability} 
\author{Andrea Lincoln\thanks{MIT, andreali@mit.edu}~\footnote{Some of this work was completed while this student was an intern at VMware. This work supported in part by NSF Grants CCF-1417238, CCF-1528078 and CCF-1514339, and BSF Grant BSF:2012338.}, ~~Adam Yedidia \thanks{MIT, adamy@mit.edu}}
\date{}
\maketitle
\thispagestyle{empty}

\setcounter{page}{0}

\begin{abstract}

We describe an algorithm to solve the problem of Boolean CNF-Satisfiability when the input formula is chosen randomly. 

We build upon the algorithms of Sch{\"{o}}ning 1999 and Dantsin et al.~in 2002.
The Sch{\"{o}}ning algorithm works by trying many possible random assignments, and for each one searching systematically in the neighborhood of that assignment for a satisfying solution.  Previous algorithms for this problem run in time $O(2^{n (1- \Omega(1)/k)})$.

Our improvement is simple: we count how many clauses are satisfied by each randomly sampled assignment, and only search in the neighborhoods of assignments with abnormally many satisfied clauses. We show that assignments like these are significantly more likely to be near a satisfying assignment. This improvement saves a factor of $2^{n \Omega(\lg^2 k)/k}$, resulting in an overall runtime of $O(2^{n (1- \Omega(\lg^2 k)/k)})$ for random $k$-SAT. 


\end{abstract}

\newpage

\section{Introduction}
\label{sec:intro}

The Boolean Satisfiability problem, known as SAT for short, is one of the best-known and most well-studied problems in computer science (e.g.  \cite{PPSZ05,PvsNPSurvey17,satAlgSurvey96,randomSATSurveyDimi, localSearchAlg}). In its general form, it describes the following problem: given an input formula $\phi$ composed of conjunctions, disjunctions, and negations of a list of Boolean-valued variables ($x_1$, $x_2$, \dots, $x_n$), determine whether or not there exists an assignment of variables to Boolean values such that $\phi$ evaluates to TRUE. SAT was the first problem shown to be NP-complete \cite{cook71,Levin73}. 

Every Boolean formula $\phi$ can be written in \emph{conjunctive normal form}, meaning that it is written as the logical conjunction of a series of disjunctive clauses. Each disjunctive clause takes as its value the logical disjunction of a series of \emph{literals}, which takes on either the same value of one of the variables $x_i$ or the negation of that value.

If we constrain the input formula $\phi$ to contain only disjunctive clauses that are of size $k$ or smaller, then that more constrained problem is known as \emph{$k$-Satisfiability}, or \emph{$k$-SAT} for short. When $k>2$ it is known to be NP-complete \cite{Karp72}. As $k$ grows, the best known runtime of the worst-case $k$-SAT problem, $O(2^{1- 1/\Theta(k)})$, grows~\cite{ETHandKSAT, PPSZ05}. 


It is well-known that in real-world Boolean Satisfiability problems, SAT solvers often vastly outperform the best known theoretical runtimes~\cite{satSolverPractice,generateSat96}. One possible explanation for this gap in performance is that most input formulas are easily solved without much computation being necessary, but that there exists a ``hard core'' of difficult-to-solve formulas that are responsible for the apparent difficulty of worst-case SAT. 

Another possible explanation for this gap in performance is that, in practice, people usually care about highly structured formulas that are much easier to solve than typical formulas---according to this explanation, there would be an ``easy core'' of tractable formulas that are responsible for the apparent simplicity of most practical SAT problems.

To try to distinguish between these two explanations, one can study \emph{random Satisfiability}: Boolean Satisfiability for which the input formula $\phi$ is chosen according to some known uniform probability distribution $D_\Phi$, and where we expect to be able to return the correct answer (satisfiable or unsatisfiable) with probability that is exponentially close to 1 in the size of the input. Random k-SAT is a very well studied problem (e.g.~\cite{randomSATSurveyDimi,threasholdSAT,machineLearnRandomSat,coja2007almost, coja2010better,ming1990probabilistic,vyas2018super}).

Typically, attention is restricted to $k$-CNF formulas whose ratio of clauses to variables is \emph{at the threshold}, meaning that the number of clauses $m$ is drawn from a Poisson distribution centered at $d_k n$, where $n$ is the number of variables and $d_k$ is a function of $k$ close to $2^k \ln(2) - \frac{1}{2}(1+\ln(2))$ \cite{threasholdSAT}. Such formulas are conjectured to be the hardest instances for a given $n$~\cite{cook1997finding, generateSat96}. It was shown by Ding, Sly, and Sun~\cite{threasholdSAT} that away from this threshold, formulas are either overwhelmingly satisfied or overwhelmingly unsatisfied, making the problem less interesting. Notably, away from this threshold one can simply return True or False based on the number of clauses and give the correct answer with high probability. We go into much greater detail about $D_\Phi$ and the threshold in Section~\ref{sec:prelims}.



Away from the threshold, polynomial-time algorithms for SAT have been found and analyzed, first by Chao and Franco~\cite{ming1990probabilistic}, and later by Coja-Oghlan et al.~\cite{coja2007almost, coja2010better}. 
Additionally, a recent result by Vyas~\cite{vyas2018super} re-analyzes the algorithm of Paturi et al.~\cite{PPZ} in the case when the input is drawn from a random distribution, and finds the algorithm to run faster on average in this case by a factor of $2^{n \Omega(\lg k)/k}$, giving a running time of $O(2^{n (1- \Omega(\lg k)/k)})$. 

We build upon the work of Sch{\"{o}}ning~\cite{randomLocalSearch} to solve random $k$-SAT in time $O(2^{n (1- \Omega(\lg^2 k)/k)})$. This represents an algorithmic improvement of $2^{n \Omega(\lg^2 k)/k}$ over the runtime of the algorithm of Paturi et al.~as analyzed by Vyas in~\cite{vyas2018super}.

\subsection{A New Algorithm \label{sec:newalg}}

In this paper, we restrict our attention to the problem of random $k$-CNF Satisfiability in the limit of large $k$, which approaches general Boolean CNF-Satisfiability. Our algorithm improves upon the previous best known algorithm for solving random $k$-SAT in the limit of large $k$, assuming that the input formulas are chosen according to a known uniform distribution.

Our algorithm improves the running time of $k$-CNF Satisfiability at the threshold by modifying the algorithm of Sch{\"{o}}ning to only explore in the neighborhood of those sampled assignments that pass an additional test.
By adding this test, we get a $2^{n\Omega(\lg^2 k)/k}$ improvement in the runtime of the algorithm. The test is simple: we count how many clauses are satisfied, and if that number is large, only then do we search in the neighborhood of the assignment. In Appendix \ref{subsec:motivation}, we provide additional motivation for why our improved running time is remarkable. 

\begin{theorem}[Main Theorem Informal]
Let $\phi$ be drawn uniformly at random from formulas at the threshold (defined formally in Section~\ref{sec:prelims}). There exists an algorithm, \HammDist~ (described in Section~\ref{sec:algorithm}), such that:
\begin{itemize}
	\item If $\phi$ is satisfiable, then with probability at least $1-3\cdot 2^{-n/(3\ln(2)2^k)}$, \HammDist~returns an assignment $\vec{a}$ that satisfies $\phi$.
	\item If $\phi$ is not satisfiable, then \HammDist~ will return False with certainty. 
	\item \HammDist($\phi$) will run 
	in time $O(2^{n(1- \Omega(\lg^2 k)/k)}).$
\end{itemize}

\end{theorem}


A key technique in the proof of our result is connecting a different distribution over inputs (the \emph{planted $k$-SAT distribution}) to the uniformly random $k$-SAT distribution. Reductions between planted $k$-SAT and random $k$-SAT have been shown in previous work as well \cite{plantedToRandom2,plantedToRandom1,vyas2018super}. In the planted $k$-SAT distribution, an assignment, $\vec{a}$, is picked first. The formula $\phi$ is selected uniformly over $k$-SAT clauses \emph{conditioned on} $\vec{a}$ satisfying those clauses. As a result, the planted distribution has a bias towards picking formulas that have many satisfying assignments, relative to the uniform distribution over all satisfiable formulas. For this reason, the planted distribution tends to generate easier-to-solve formulas $\phi$ than the uniform distribution~\cite{plantedSATPoly}. We also find that the planted distribution is more easily analyzed.

It would be possible, and simpler, to analyze our algorithm only in the planted distribution over formulas. This would \emph{not}, however, correspond to a complete analysis of the algorithm in the random case. In this work, we begin by analyzing the performance of our algorithm when run on inputs drawn from the planted distribution. However, in Lemma~\ref{lem:plantedToAllReal}, we show that algorithms with a sufficiently low probability of failure in the planted distribution over input formulas continue to have a low probability of failure in the uniform distribution over input formulas. Similar reductions have been proven in previous work~\cite{plantedToRandom2, plantedToRandom1,vyas2018super}. 


The bulk of the analysis of our algorithm presented in this paper will focus on four quantities. Informally:

\begin{enumerate}
    \item The \emph{true positive rate} $p_{TP}$ describes the fraction of all assignments that are both close to a satisfying assignment in Hamming distance and satisfy a large number of clauses.
    \item The \emph{false negative rate} $p_{FN}$ describes the fraction of all assignments that are close to a satisfying assignment in Hamming distance, but do not satisfy a large number of clauses.
    \item The \emph{false positive rate} $p_{FP}$ describes the fraction of all assignments that are not close to any satisfying assignment in Hamming distance, but satisfy a large number of clauses.
    \item The \emph{true negative rate} $p_{TN}$ describes the fraction of all assignments that are neither close to any assignment in Hamming distance, nor satisfy a large number of clauses.
\end{enumerate}
    
By showing that the true positive rate is large enough relative to the false positive rate, we show that we do not too often perform a ``useless search,'' i.e.~one that will not find a satisfying assignment. And by showing that the true positive rate is large enough relative to the total number of possible assignments, we show that we eventually do find a satisfying assignment without needing to take too many samples. See Fig.~\ref{fig:binom} for an illustration of these concepts.

To show that our algorithm achieves the desired runtime, we must demonstrate two things. First, we must show that false positives are sufficiently rare; in other words, conditioned on an assignment passing our test, it is sufficiently likely to be a small-Hamming-distance assignment. We prove this in Section~\ref{sec:falsePositive}. Second, we must show that true positives are sufficiently common; in other words, conditioned on an assignment being close in Hamming distance to a satisfying assignment, it is sufficiently likely to pass our test. We prove this in Section~\ref{sec:falseNegative}.


We also note that our algorithm can potentially be used as the seed for a worst-case algorithm. Informally, the correctness of the analysis in this paper depends only on the false positive and false negative rates being sufficiently low. As long as the inputs are guaranteed to come from a family of formulas for which this is the case, our algorithm will work even in the worst case. Or, to put it another way, to build a working worst-case algorithm using our algorithm as a template, one may now restrict one's attention to solving input formulas for which assignments in the neighborhood of the solution do \emph{not} have an abnormally-high number of satisfied clauses; our algorithm can solve the others.


\begin{figure}
    \centering
    \subfloat[The full histogram.]{\includegraphics[width=3in]{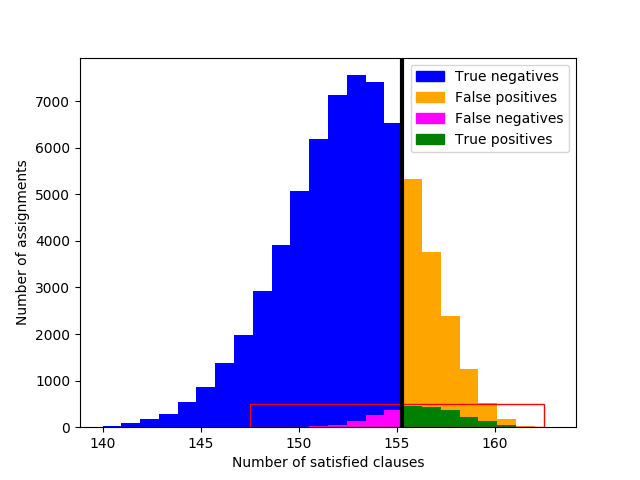}}
    \subfloat[A magnification of the red region shown on the left.]{\includegraphics[width=3in]{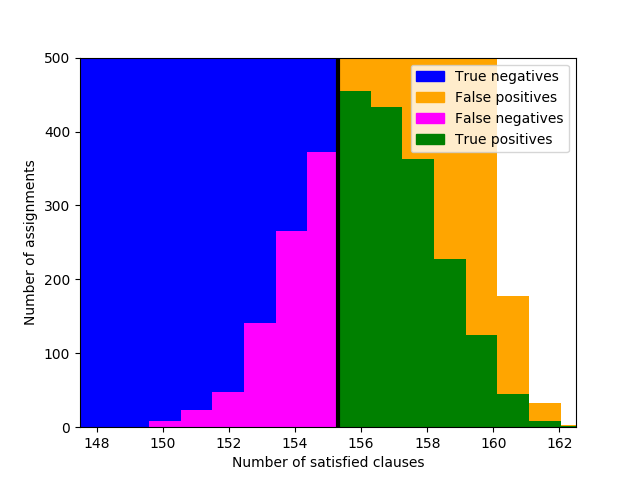}}
    \caption{A histogram of how many clauses are satisfied by every possible assignment. In this example, there are $n=16$ variables, $m = 163$ clauses, and $k = 4$ literals per clause. For the example, we take $T = 155.5$ to be the clause-satisfaction threshold above which we explore further, and $\alpha n = 4$ to be the small-Hamming-distance threshold at which the exhaustive search algorithm finishes. (In actual runs of the algorithm, both of these parameters are selected more conservatively; we chose these parameters for clarity.)}
    \label{fig:binom}
\end{figure}

\subsection{Previous work}
Satisfiability and $k$-SAT have been thoroughly studied. We will cover some of the previous work in the area, focusing on the Random $k$-SAT problem. 

\paragraph{Structural Results About Random $k$-SAT}

To make the study of the Satisfiability of random formulas interesting, it is important to choose the probability distribution over formulas $D_\phi$ judiciously. In particular, $D_\phi$ must contain formulas where the ratio of Boolean variables to disjunctive clauses is such that the resulting formulas are neither overwhelmingly satisfiable, nor overwhelmingly unsatisfiable. Let $n$ be the number of variables, and $m$ be the number of clauses. If $n \gg m$, then nearly all formulas chosen uniformly from $D_\phi$ will be satisfiable; if $m \gg n$, then nearly all formulas will be unsatisfiable. In order for the problem of correctly identifying formulas as satisfiable or unsatisfiable to be nontrivial, we must choose $m$ and $n$ to be at the right ratio. Throughout this paper we will refer to the ratio of $m$ to $n$ as the \emph{density} of a formula. 

In work by Ding, Sly and Sun~\cite{threasholdSAT}, it was shown that a sharp threshold exists between formulas which are satisfied with high probability and those that are unsatisfied with high probability. More precisely, they describe what happens when the number of clauses $m$ is drawn from a Poisson distribution with mean $d_k n$. When the number of clauses drawn is below $(d_k - \epsilon)n$, only an exponentially small fraction of formulas will be unsatisfiable; when the number of clauses drawn is greater than $(d_k + \epsilon)n$, only an exponentially small fraction of formulas will be satisfiable. This holds true for any $\epsilon > 0$ constant in $n$.




\paragraph{Previous Average-Case $k$-SAT Algorithms}

Feldman et al.~studied planted random $k$-SAT and found that given $m = \Omega(n \lg n )$ clauses, the planted solution can be determined using statistical queries \cite{plantedSATPoly}. Feldman et al.~also conjecture that planted $k$-SAT is easier than random $k$-SAT more generally. Previous work has shown a connection between algorithms that work in the planted distribution and algorithms that work in the random distribution~\cite{plantedToRandom1,plantedToRandom2,vyas2018super}.
An algorithm was found by Valiant which runs in time $O(2^{n(1-\Omega(\log(k))/k)})$ at the threshold~\cite{Val}, improving upon PPSZ~\cite{PPSZ05}. Additionally, Vyas~\cite{vyas2018super} obtained the same runtime by re-analyzing the algorithm of Paturi et al.~\cite{PPZ} in the random case. 

\paragraph{Worst-Case $k$-SAT Algorithms}


The previously best-known worst-case $k$-SAT algorithms for large $k$ are due to Paturi et al.~who get a running time of $O(2^{n(1-\Omega(1)/k)})$ \cite{PPSZ05}.
Previous work by Sch{\"{o}}ning gave an algorithm to solve $k$-SAT in the worst case with an expected running time of $O(2^{n(1-\Omega(1)/k)})$ (but with a worse constant in the $\Omega(1)$) ~\cite{randomLocalSearch}. Our algorithm is a modification of Sch{\"{o}}ning's. Their algorithm runs by choosing an assignment at random, and searching in the immediate neighborhood of that assignment by repeatedly choosing an unsatisfied clause and flipping a variable in that clause to satisfy it. They perform the search near the randomly chosen assignment via an exhaustive search. Their algorithm is an improvement over a naive brute-force algorithm because of the savings that result from only considering variable-flips that could possibly cause the formula to become satisfied (rather than also exploring variable-flips that can't possibly be helpful).

\section{Preliminaries}
\label{sec:prelims}
In this section we will give the definition of random $k$-CNF Satisfiability (random $k$-CNF SAT) at the threshold. We additionally present definitions of several important distributions and functions that are used later in the paper. 

\paragraph{Notation for This Paper}

We use $x \sim D$ to indicate that $x$ is a random value drawn from the distribution $D$.

We use $f(x) = O^*\left(g(x)\right)$ to denote that there exists some constant $c$ such that $f(x) = O\left(g(x) x^c\right)$. So, to say it another way, $f(x)$ grows at most as quickly as $g(x)$, ignoring polynomial and constant factors.

We often use ``small-Hamming-distance assignment'' to mean an assignment that is a small Hamming distance from a satisfying assignment. 

\paragraph{Definitions for This Paper}


\begin{definition}
	Let $D_{replace}(n, k)$ be the uniform distribution over \textit{clauses} on $k$ variables where those $k$ variables are chosen with \textit{replacement} (e.g. $(\bar{x} \vee x \vee y )$ would be a valid clause). Under this definition, there are $n^k2^k$ possible clauses. 
\end{definition}

\begin{definition}
The density of a formula $\phi$ with $m$ clauses and $n$ variables is $m/n$.
\end{definition}

We will now define the \textit{satisfiability threshold}. Informally, this is a density of clauses such that formulas drawn from below this threshold are with high probability (whp) satisfied, and those formulas drawn from above the threshold are whp unsatisfied. 

\begin{definition}
The \textit{satisfiability threshold}, $d_k$, is a ratio of clauses to variables such that for all $\eps>0$:
\begin{itemize}
	\item If $m$ is drawn from $Pois[(d_k-\eps)n]$, the Poisson distribution with mean $(d_k-\eps)n$, and $\phi$ is formed by picking $m$ clauses independently at random from $D_{replace}(n, k)$, then $\phi$ is whp \textbf{satisfied}. 
	\item If $m$ is drawn from $Pois[(d_k+\eps)n]$, the Poisson distribution with mean $(d_k+\eps)n$, and $\phi$ is formed by picking $m$ clauses independently at random from $D_{replace}(n, k)$, then $\phi$ is whp \textbf{unsatisfied}. 
\end{itemize} 
\end{definition}

\begin{definition}
	Let $d_k$ be the density of clauses such that SAT is at its threshold.
\end{definition}

Note that it is not immediate that a satisfiability threshold exists for any given $k$. 
However, Jian Ding, Allan Sly, and Nike Sun showed that this threshold exists for sufficiently large $k$~\cite{threasholdSAT}. It has been proven that 
$$ 2^k \ln(2) - \frac{1}{2}(1+\ln(2)) -\epsilon_k \leq d_k \leq 2^k \ln(2)  - \frac{1}{2}(1+\ln(2)) +\epsilon_k ,$$ 
where $\epsilon_k$ is a term that tends to $0$ as $k$ grows \cite{boundDensitySAT,boundingSATDensity2}. It follows that there exists a large enough $k$ such that $ 2^k \ln(2) - 1 \leq d_k$. Also, there exists a large enough $k$ such that $ 2^k \ln(2) \geq d_k$.

This density determines the distribution over the number of clauses put in the formula. Specifically, $m$ is drawn from $Pois[d_kn]$. However, we can say that with high probability the number of clauses is nearly $d_kn$ (see Lem.~\ref{lem:boundingM}).

For many proofs it is convenient to assume $k$ is large (e.g. when $k$ is large, $2^k>10k$ not just asymptotically but also numerically). 
We will now define $\kstr$. It will be a value such that $k = \kstr$ is large enough that both $d_k$ is known to be close to $2^k \ln(2) - \frac{1}{2}(1+\ln(2))$ and large enough for our proofs that depend on $k$ being large. 

\begin{definition}
	Let $\eps_k = |d_k - 2^k \ln(2) + \frac{1}{2}(1+\ln(2))|$.
\end{definition}

\begin{definition}
Let $k_{\eps}$ be the minimum value such that for all $k\geq k_{\eps}$ we have that $\eps_k < \frac{1+\ln(2)}{2}$.
\end{definition}

\begin{definition}
Let $\kstr = \max \left(\kval , k_{\eps} \right)$.
\label{def:kstar}
\end{definition}

Our choice of $\kval$ in the above is somewhat arbitrary. When $k\geq \kval$ the proofs in section \ref{sec:lotsOfSAT} are simpler, so we analyze our core algorithm in that regime. 

\begin{lemma}
If $\eps_k < \frac{1+\ln(2)}{2}$ then
$ Pr[m>(d_k+1)n ] + Pr[m<  (d_k-1)n ]\leq2 \cdot 2^{-n/(3\ln(2)2^k)} $. 
\label{lem:boundingM}
\end{lemma} 
\begin{proof}
We apply the multiplicative form of the Chernoff bound. We have that $(d_k+1)n/(d_kn) = 1+ 1/d_k$. We also have that $(d_k+1)n/(d_kn) = 1- 1/d_k$. This gives us
$$Pr[m>(d_k+1)n ] + Pr[m<  (d_k-1)n ]\leq 2^{-(d_k)^{-2}d_kn/3}+2^{-(d_k)^{-2}d_kn/2}.$$
Which means 
$$Pr[m>(d_k+1)n ] + Pr[m<  (d_k-1)n ]\leq 2^{-n/(3d_k)}+2^{-n/(2d_k)}.$$
$$Pr[m>(d_k+1)n ] + Pr[m<  (d_k-1)n ]\leq 2 \cdot 2^{-n/\left(3(\ln(2)2^k-\frac{1}{2}(1+\ln(2))\right)}.$$
$$Pr[m>(d_k+1)n ] + Pr[m<  (d_k-1)n ]\leq 2 \cdot 2^{-n/(3\ln(2)2^k)}.$$
\end{proof}


It follows that if our algorithm works efficiently for all values of $m \in [(d_k-1)n, (d_k+1)n]$, then it works with high probability at the threshold.

Below are some definitions used in later sections.


\begin{definition}
	Let $D_{\Phi}(n,k)$ be the distribution over formulas $\phi$ where all clauses are chosen independently from $D_{replace}$ and the number of clauses is chosen from a Poisson distribution with mean $d_kn$.
\end{definition}


\begin{definition}
Let $D_R(m,n,k)$ be the distribution over formulas $\phi$ where all $m$ clauses are chosen independently from  $D_{replace}$.
\end{definition}


\begin{definition}
Let $D_S(m,n,k)$ be the uniform distribution over \emph{satisfied} formulas $\phi$ where all $m$ clauses are chosen from $D_{replace}$.
\end{definition}

\begin{definition}
Let $\DPc(n, k, \vec{a})$ (which we refer to as ``the planted-clause distribution'') be the uniform distribution over the $(2^k-1)n^k$ clauses $c$ which are satisfied by $\vec{a}$.
\end{definition}

\begin{definition}
Let $\DPa(m, n, k, \vec{a})$ (which we refer to as ``the planted distribution'') be the distribution over formulas $\phi$ where every clause is picked IID from $\DPc(n, k, \vec{a})$. Note that this is equivalent to the uniform distribution over formulas $\phi$ which are satisfied by $\vec{a}$ and where all $m$ clauses are in the support of $D_{replace}$.
\end{definition}

\begin{definition}
Let $U_{\vec{a}}(n)$ be the uniform distribution over assignments of length $n$, $\{0, 1\}^n$.
\end{definition} 

\begin{definition}
Let $\ncs(\phi, \vec{v})$ be the number of clauses in $\phi$ satisfied by the assignment $\vec{v}$.
\end{definition}

\begin{definition}
Let $\ncus(\phi, \vec{v})$ be the number of clauses in $\phi$ left unsatisfied by the assignment $\vec{v}$.
\end{definition}

\section{Algorithm}
\label{sec:algorithm}
We will describe our algorithm for random $k$-SAT in this section.

Informally, our algorithm works as follows. Given an input formula, we will sample many randomly-chosen assignments. On those that have a high number of satisfied clauses, we will run the deterministic algorithm for finding a satisfying assignment given an assignment that is within a Hamming distance of at most $\alpha n$ of that satisfying assignment (i.e. a small-Hamming-distance assignment). 

Unsurprisingly, in the average case, small-Hamming-distance assignments satisfy more clauses than random assignments\footnote{Consider changing one variable's assignment at random; in this case, almost all clauses will remain satisfied. This phenomenon persists even when we flip several variables at once.}. In fact, for many choices of criterion there will be a discrepancy between the values achieved by small-Hamming-distance assignments and random assignments. Lemma \ref{lem:generalSmallHDf}, which characterizes this discrepancy, is general enough to be applied immediately to analyzing algorithms that make use of any clause-specific criterion. 

We note the following from previous work:

\begin{lemma}[Small Hamming Distance Search \cite{localSearchAlg}]
	There is a deterministic algorithm 
	\smlHDAlg($\phi,\vec{v}$) which given 
	\begin{itemize}
		\item a k-CNF formula $\phi$ on $m$ clauses and $n$ variables, and
		\item an assignment $\vec{v}$ which has Hamming distance $\alpha n$ from a true satisfying assignment $\vec{a}^*$,
	\end{itemize}
	  will return a satisfying assignment within Hamming distance $\alpha n$ of $\vec{v}$ if one exists in $k^{\alpha n}$ time.
\end{lemma}

This algorithm simply takes the assignment $\vec{a}$ and branches on the first unsatisfied clause, trying all possible variable flips. For each assignment resulting from these possible variable flips, the algorithm repeats the process in what is now the first unsatisfied clause, until it either finds a satisfying assignment or has searched $\alpha n$ flips from the original assignment. This will deterministically yield a satisfying assignment, should one exist, within a Hamming distance of $\alpha n$ of the original assignment. 

So, if we find a small-Hamming-distance assignment and run \textsc{SAT-from-$\alpha$-Small-HD}($\phi,\vec{a}$) on this assignment, we are guaranteed to find the satisfying assignment. Therefore, we could randomly sample points until we expect to find an assignment at Hamming distance $\alpha n$ from the satisfying assignment (call this an $\alpha$-small-Hamming-distance assignment). This is indeed what Sch{\"{o}}ning's algorithm does for $\alpha = \Theta(1/k)$ \cite{randomLocalSearch}. 

A general class of improvements to this algorithm work by running \smlHDAlg($\phi,\vec{a}$) on only a cleverly-chosen subset of these sampled assignments. In our case, we choose this set to be assignments that satisfy an unusually large number of clauses, but in principle one could use any membership criterion for this set. 

Let $M$ be the runtime of the membership test for the set of assignments, and let $p_{TP}$, $p_{FP}$, $p_{FN}$, and $p_{TN}$ represent the fraction of assignments that are true positives, false positives, false negatives, and true negatives respectively. Here, just as in Section~\ref{sec:newalg}, we use ``positive'' or ``negative'' to mean an assignment that passes or doesn't pass the test for membership, respectively. The truth or falsehood of that positive or negative represents whether or not that assignment actually has a satisfying assignment within small Hamming distance.

We will have to draw samples until we would have found a satisfying assignment with high probability were one to exist. Next, we will have to run \smlHDAlg($\phi,\vec{a}$) at least once to find the satisfying assignment itself. Finally, we will have to run it once more for every false positive we find. Hence, the the generalized running time of this class of algorithms is 

\begin{equation}
O^*\left(\frac{M}{p_{TP}} + k^{\alpha n} + \left(\frac{p_{FP}}{p_{TP}}\right) k^{\alpha n}\right).
\label{eq:general_rt}
\end{equation}

This general formula is a powerful tool for analyzing the runtimes of algorithms from this class. For example, if we apply it to analyzing the algorithm of~\cite{localSearchAlg}, i.e. the special case where the test we use always returns a positive, we see that the third term in Equation~(\ref{eq:general_rt}) dominates, and that $p_{FP} \approx 1$ and $p_{TP} \approx \frac{\binom{n}{\alpha n}}{2^n}$, giving an overall expected runtime of $O^*(2^{n(1-\Theta(\frac{1}{k}))})$ when we choose $\alpha = \Theta(\frac{1}{k+1})$. In Appendix \ref{subsec:kalphaVSexaustive} we discuss a different deterministic search algorithm with a slightly improved runtime (yielding no relevant improvement on the runtime of the overall algorithm for our analysis).

\newcommand{\LocalSearch}{\textsc{DantsinLS}}

Our algorithm presents improvements for large $k$, but for small $k$ we will simply use the previous algorithm of Dantsin et al \cite{localSearchAlg}. 
\begin{lemma}[Algorithm for Small $k$ \cite{localSearchAlg}]
For $k \leq \kstr$ there exists a deterministic algorithm, \LocalSearch, that solves $k$-SAT in the worst case in time $2^{n(1-\gamma)}$ for some constant $\gamma>0$ \cite{localSearchAlg}.
\end{lemma}

We will now give pseudocode for the \HammDist~ algorithm in Algorithm \ref{alg:randomSatAlgo}. 
Let \textsc{NumClausesSAT}($\phi$,$\vec{a}$) return the number of clauses in $\phi$ satisfied by the assignment $\vec{a}$. In Appendix \ref{subsec:regularizeVStest}, we describe a different set of concepts with which the algorithm can be understood.

\begin{algorithm*}[h]
    \SetAlgoLined
	\HammDist($\phi$):\\
	\nl \If{$k<\kstr$}
	{\nl \textbf{return \LocalSearch($\phi$)}\\}

	\nl Initialize $S$ to the empty set. \\
	\Comment{\emph{For $k\geq \kstr$ we run our variant of local search:}}\\
	\nl \For{$i \in [0,n^2 \cdot 2^n/\binom{n}{\alpha n}]$ }{
	\nl	Sample an assignment $\vec{a}$ uniformly at random from $\{0,1\}^n$.\\
		\Comment{\emph{Only keep assignments which satisfy abnormally many clauses.}}\\
	\nl	\If{\ncs($\phi$,$\vec{a}$) $\geq$ $(1-\frac{1-(1-\alpha)^{2k}}{2^k-1})m$}{
	\nl		Add $\vec{a}$ to $S$. \\
	\nl		\If{$|S| > 4 n^3 2^n/\left(\binom{n}{\alpha n} k^{\alpha n}\right) +1$}{\nl \textbf{return} False}
	\nl		Run \smlHDAlg($\phi,\vec{a}$).\\
	\nl		If an assignment was found, return it.}}
	\nl \textbf{return} False\\
	\caption{\HammDist($\phi$)}
	\label{alg:randomSatAlgo}
\end{algorithm*}

Note that our algorithm as stated is non-constructive due to our use of the constant $\kstr$. Other than this constant, our algorithm is explicit. While $\kstr$ is known to be constant~\cite{boundingSATDensity2}, its exact value is currently unknown. We note in Section~\ref{sec:conclusion} that finding the value of $\kstr$ is an open problem which, if solved, would make our algorithm constructive.


\subsection{Correctness and Running Time}


We will include the theorem statement of correctness and running time here. Its proof depends on bounds on the false positive rate and the true positive rate, which we prove in later sections. In particular, we show in Section~\ref{sec:falsePositive}  that conditioned on an assignment passing the test, it is sufficiently likely to be an $\alpha$-small-Hamming-distance assignment. We additionally show  in Section~\ref{sec:falseNegative} that conditioned on an assignment being an $\alpha$-small-Hamming-distance assignment, it is sufficiently likely to pass the test.


Note that much of our probability of returning the wrong value comes from our bounds on the probability that we are drawing a formula with length $m < (d_k-1)n$. If we knew $m$ to be fixed and greater than $(d_k-1)n$, we would have a lower error probability.

We will show that \HammDist($\phi$) has one-sided error and returns the correct answer with high probability. Note that it returns the correct answer with high probability even conditioned on the input being unsatisfied or satisfied. We use Theorem \ref{thm:runningTimeIsGood} to bound the false positive rate and use Lemma \ref{lem:algIsCorrect} to bound the true positive rate, which gives us the desired result.

In the theorem that follows, we choose $\alpha$ such that $\alpha n$ is an integer. Specifically, we choose:

$$\alpha = \frac{\lfloor \frac{\lg(k)}{16k} n \rfloor}{n}.$$

Note that when we choose $\alpha$ to take on this value, it will always lie in the range $\frac{\lg(k)}{20k} \le \alpha \le \frac{\lg(k)}{16k}$ for large $n$.

\begin{theorem}
Assume $\phi$ is drawn from $D_{\Phi}(n,k)$. Let $\alpha = \frac{\lfloor n\lg(k)/(16k) \rfloor}{n}$.

Conditioned on there being at least one satisfying assignment to $\phi$, \HammDist($\phi$) will return some satisfying assignment with probability at least $ 1-3\cdot 2^{-n/(3\ln(2)2^k)}$.

Conditioned on there being no satisfying assignment to $\phi$, \HammDist($\phi$) will return False with probability $1$. 

\HammDist($\phi$) will run 
in time 
$$O\left(2^{n (1- \Omega(\lg^2(k)/k)}\right).$$

	\label{thm:HammDistCorrectFast}
\end{theorem}
\begin{proof}
Proof given in Section \ref{sec:putItTogether}. 
\end{proof}

%


\section{Bounding the False Positive Rate}
\label{sec:falsePositive}

To show that our algorithm runs efficiently, we must show that we do not too often find assignments that have an abnormally large number of satisfied clauses but are not close to a satisfying assignment in Hamming Distance, leading to wasted effort. To do this, we simply argue that few enough assignments pass our test that this is not an issue. This bounds the sum $p_{FP} + p_{TP}$ as defined in Sections~\ref{sec:intro}  and~\ref{sec:algorithm}, which of course is itself an upper bound on $p_{FP}$.


\begin{lemma}
	Let $\epsilon > 0$ and $0 < \alpha < 1$ be constants such that $2^k(1-\alpha)^{2k} > \frac{1+\epsilon}{\epsilon}$, and let $\tilde{T}=\frac{1-(1-\alpha)^{2k}}{2^k-1}m$.
    

    Given a  $\phi$ drawn from $D_R(m,n,k)$, then with probability $$\geq 1-\exp \left(-\frac{1}{4(1+\epsilon)^2\cdot 2^k}(1-\alpha)^{4k}m \right) 2^n$$ 
    
    
	we have that the number of assignments $\vec{v}$ such that
    $$\ncus(\phi, \vec{v}) \le \tilde{T}$$
    is less than $\exp \left(-\frac{1}{4(1+\epsilon)^2\cdot 2^k}(1-\alpha)^{4k}m \right) 2^n$.
	 \label{lem:tailBoundGeneralM}
\end{lemma}
\begin{proof}	
	

    Given a uniformly random assignment $\vec{u}$, a $\phi$ chosen from $D_R(m,n,k)$ will have $m$ clauses, each of which will be chosen uniformly at random from the $(2n)^k$ possible clauses. Each literal of each clause will be unsatisfied with independent probability $\frac{1}{2}$; thus each of the $m$ clauses will be unsatisfied with independent probability $\frac{1}{2^k}$. Thus, the value of $\ncus(\phi, \vec{u})$ will be drawn from a binomial distribution with probability $p=\frac{1}{2^k}$ and number of samples $m$. We want to bound from above the probability of a draw with abnormally few unsatisfied clauses, 
    
    $$\pprom = Pr_{\phi \sim D_R(m, n, k), \vec{u} \sim U_{\vec{a}}(n)}[\ncus(\phi, \vec{u}) \le \tilde{T}].$$
    
    The mean $\mu$ of this binomial distribution is given by $\mu = \frac{m}{2^k}$. We can apply known bounds on the probability of extremal draws from binomial distributions. Note that $\tilde{T} = \mu(1-\delta)$, where:
    
    $$\delta = \frac{2^k(1-\alpha)^{2k} - 1}{2^k - 1}.$$
    
    Using the multiplicative Chernoff bound, we have that

    $$\pprom \le e^{\frac{-\delta^2 \mu}{2}}.$$
    
    which implies
    
    $$\pprom \le e^{-m\frac{(1-\alpha)^{4k}}{2(1+\epsilon)^2\cdot2^k}}$$
    
    for constant $\epsilon > 0$. 
    
    
	
	

    The statement just proved is related to, but not precisely the same as, the lemma's statement. The former counts extremal \emph{pairs} of formulas and assignments, whereas the latter counts formulas with an extremal number of extremal assignments. In what follows, we use the statement above to prove the lemma's statement.

    Let $(q_1, q_2)$ be a tuple with the following property: $q_2$ is the probability that there exist at least $q_1 2^n$ assignments $\vec{w}$ such that $\ncus(\phi', \vec{w}) < \tilde{T}$, when $\phi' \sim D_R(m,n,k)$. Then, by the definition of $\pprom$, we know that $\pprom \ge q_1q_2$. If we let $q_1 = q_2 = \sqrt{\pprom}$, then we have that the probability that there exists at least $2^n\sqrt{\pprom}$ assignments $\vec{w}$ such that $\ncus(\phi', \vec{w}) < \tilde{T}$ is at most $\sqrt{\pprom}$. 
	
	 Thus the fraction of formulas with at least 
	 $$\exp \left(-\frac{1}{4(1+\epsilon)^2\cdot 2^k}(1-\alpha)^{4k}m \right) 2^n$$
	  assignments above the threshold is 
	  $$\leq \exp \left(-\frac{1}{4(1+\epsilon)^2\cdot 2^k}(1-\alpha)^{4k}m \right)$$
     for constant $\epsilon > 0$. 
      
\end{proof}

\begin{corollary}
	Let $k\geq \kstr$.\\
	Let $\epsilon > 0$ and $0 < \alpha < 1$ be constants such that $2^k(1-\alpha)^{2k} > \frac{1+\epsilon}{\epsilon}$. 
	Let $m = (\ln(2)2^k-\gamma)n$ and $\gamma \in [0, \ln(2)2^k/4]$.

	Given a  $\phi$ drawn from $D_R(m,n,k)$, then with probability greater than $$1-\exp \left(-\frac{\ln(2)}{8(1+\epsilon)^2}(1-\alpha)^{4k}n \right)$$ 
	we have that the number of assignments $\vec{v}$ such that 
$\ncs(\phi, \vec{v}) \ge \left(1-\frac{1-(1-\alpha)^{2k}}{2^k-1}\right)m$ 
    is less than
	$$\exp \left(-\frac{\ln(2)}{8(1+\epsilon)^2}(1-\alpha)^{4k}n \right)2^n.$$
	\label{cor:FPNumberbound}
\end{corollary}
\begin{proof}
Plug in the bounding values of $m$ to Lemma \ref{lem:tailBoundGeneralM}. 

The condition of $\ncs(\phi, \vec{v}) \ge (1-\frac{(1-(1-\alpha)^{2k}}{2^k-1})m$ 
is equivalent to $\ncus(\phi, \vec{v}) \le \tilde{T}$ from Lemma \ref{lem:tailBoundGeneralM}. 
\end{proof}

\begin{lemma}
	Let $k\geq \kstr$.\\
	Let $\epsilon > 0$ and $0 < \alpha < 1$ be constants such that $2^k(1-\alpha)^{2k} > \frac{1+\epsilon}{\epsilon}$. \\
	Let $m = (\ln(2)2^k-\gamma)n$ and $\gamma \in [0, \ln(2)2^k/4]$.\\    
    Let $$p_{\mathrm{pos}} = \exp \left(-\frac{\ln(2)}{8(1+\epsilon)^2}(1-\alpha)^{4k}n \right).$$
    Let $c_s = n^2 2^n / \binom{n}{\alpha n}$, the total number of samples we take in \HammDist.

Let $\mu_{\mathrm{pos}} = c_s p_{\mathrm{pos}}$.
	
Run algorithm \HammDist($\phi$) on a randomly selected $\phi$ drawn from $D_R(m,n,k)$. Then, with probability at least
$$1-\exp \left(-\frac{\ln(2)}{8(1+\epsilon)^2}(1-\alpha)^{4k}n \right)-e^{-n \mu_{\mathrm{pos}}/3},$$
 we have that the size of $S$, as defined in Algorithm~\ref{alg:randomSatAlgo}, is less than $\leq 2n\mu_{\mathrm{pos}}$.
 
 \label{lem:boundNumSamplesGeneral}
\end{lemma}
\begin{proof}
Let the threshold  be $T=(1-\frac{1-(1-\alpha)^{2k}}{2^k-1})m$. By Corollary \ref{cor:FPNumberbound}, we have that with probability at least $1-\exp \left(-\frac{\ln(2)}{8(1+\epsilon)^2}(1-\alpha)^{4k}n \right)$, the total number of assignments $\vec{v}$ such that $\ncs(\phi, \vec{v}) \ge T$ is at most $\exp \left(-\frac{\ln(2)}{8(1+\epsilon)^2}(1-\alpha)^{4k}n \right) \cdot 2^n$.

We now want to bound the probability that too many of the sampled assignments are above the threshold. 

We know that a random assignment has probability at most $p_{\mathrm{pos}}$
of being above the threshold. And we will draw $c_s$ assignments.

Therefore, the mean of the number of drawn assignments above the threshold is $\mu \leq \mu_{\mathrm{pos}}.$

Now we can apply the multiplicative Chernoff bound (note that $(2n\mu_{\mathrm{pos}}-\mu)/\mu \geq 1$):

$$Pr[|S|>2n\mu_{\mathrm{pos}}]\leq e^{-(2x-\mu)\mu / (3\mu )} = e^{-(2n\mu_{\mathrm{pos}}-\mu)/3} \leq e^{-n\mu_{\mathrm{pos}}/3}.$$
\end{proof}

\begin{lemma}
	Let $k\geq \kstr$.\\
	Let $\alpha= c\lg(k)/k$ for some constant $c<1/8$.
    
    Let $\epsilon > 0$ be a constant such that $2^k(1-\alpha)^{2k} > \frac{1+\epsilon}{\epsilon}$.
	
	Let $m = (\ln(2)2^k-\gamma)n$ and $\gamma \in [0, \ln(2)2^k/4]$.   	
	
	Run algorithm \HammDist($\phi$, $n$, $m$) on a randomly selected $\phi$ drawn from $D_R(m,n,k)$. Then, with probability at least
	$$1-2 \cdot 2^{-n/(8(1+\epsilon)^2 k^{8c})},$$
	we have that the size of $S$ is less than or equal to $4 n^3 2^n/\left(\binom{n}{\alpha n} k^{\alpha n}\right)$.
	\label{lem:boundNumSamplesOurAlpha}
\end{lemma}
\begin{proof}
Plug in $\alpha = c\lg(k)/k$ into Lemma \ref{lem:boundNumSamplesGeneral}.

Note that $x > 2^n/2$ and 
\begin{align}
\exp \left(-\frac{\ln(2)}{8(1+\epsilon)^2}(1-\alpha)^{4k}n \right) &\leq \exp \left(-\frac{\ln(2)}{8(1+\epsilon)^2}(1/4)^{8\alpha k}n \right)\\
&\leq 2 ~ \hat{} \left(-\frac{1}{8(1+\epsilon)^2}n/k^{8c} \right)\\
\end{align}
so the probability that the event happens is at least $1-2 \cdot 2^{-n/((8(1+\epsilon)^2 k^{8c})}$. 

The size of the set is at most $2^n/(\binom{n}{\alpha n}) \cdot 2 ~ \hat{} \left(-\frac{1}{8(1+\epsilon)^2}n/k^{8c} \right)$. 

If $c<1/8$ and constant, then for large enough $k$,
$$k^{\alpha n} = 2^{n \lg(k)^2/(ck)} \leq 2 ~ \hat{} \left(\frac{1}{8(1+\epsilon)^2}n/k^{8c} \right).$$

So the size of $S$ is less than or equal to $4 n^3 2^n/\left(\binom{n}{\alpha n} k^{\alpha n}\right)$.
\end{proof}

\begin{corollary}
	Let $m = (\ln(2)2^k-c)n$.	
	
	The probability that $\phi$ drawn from $D_R(m,n,k)$  has at least one satisfying assignment is at least $$\frac{1}{n}2^{- n/e^{k/e^2}}e^{-3cn/2^k}.$$
	\label{cor:numSatIsHigh}
\end{corollary}
\begin{proof}
	Proof given in Section \ref{sec:lotsOfSAT}.
\end{proof}

Theorem~\ref{thm:runningTimeIsGood} bounds the probability that our algorithm stops prematurely and returns False on satisfiable assignments, thereby giving the wrong answer (see line 8 of Algorithm~\ref{alg:randomSatAlgo}). It considers only values of $m$ within a range that $D_\Phi(n, k)$ yields with high probability. Within that range of $m$, we can combine our total bound over the number of formulas with an abnormally large number of false positive assignments (as given by Lemma~\ref{lem:boundNumSamplesOurAlpha}) with our bound on the total number of satisfiable formulas, which we call $p_{SAT}(n, k)$ (as given by Corollary~\ref{cor:numSatIsHigh}). 

\begin{theorem}
	Let $\alpha = \frac{\lfloor \frac{\lg(k)}{16k} n \rfloor}{n}$, and let the formula $\phi \sim D_{\Phi}(n,k)$ be the input to the problem. When running $\HammDist(\phi)$  on \RandomSat instances at the threshold, conditioned on the formula being satisfiable, the probability that $S > 4 n^3 2^n/\left(\binom{n}{\alpha n} k^{\alpha n}\right)$ is at most $2\cdot 2^{-n/(3\ln(2)2^k)}$.
	\label{thm:runningTimeIsGood}
\end{theorem}
\begin{proof}
Let $\epsilon > 0$ be a constant such that $2^k(1-\alpha)^{2k} > \frac{1+\epsilon}{\epsilon}$. If $k>\kstr$ and $\alpha \ge (1/20)\lg(k)/k$, then any $\epsilon > 10^{-4}$ works. 
	
Using Lemma \ref{lem:boundNumSamplesOurAlpha} when setting $c=1/16$, we get that the size of the set $S$ is at most $ 4 n^3 2^n/(\binom{n}{\alpha n} k^{\alpha n})$ with probability at least $ 1-2 \cdot 2^{-n/(32\sqrt{k})}$ if $m$ in is in the range $(d_k - 1)n \le m \le (d_k + 1)m$. 

We can use Corollary \ref{cor:numSatIsHigh} to bound the probability that $S > 4 n^3 2^n/\left(\binom{n}{\alpha n} k^{\alpha n}\right)$  conditioned on the formula drawn from $D_{\Phi}(n,k)$ being satisfiable. Let $p_{SAT}(n,k)$ be the probability a formula drawn from $D_{\Phi}(n,k)$ is satisfiable conditioned on  $m$ being in the range $(d_k - 1)n \le m \le (d_k + 1)m$. Then, the probability that $|S| > 4 n^3 2^n/(\binom{n}{\alpha n} k^{\alpha n})$ is at most $2 \cdot 2^{-n/(32\sqrt{k})}/p_{SAT}(n,k)$, conditioned on the formula being satisfiable and $m$ being in the range $(d_k - 1)n \le m \le (d_k + 1)m$. 

We can use Corollary \ref{cor:numSatIsHigh} to bound $p_{SAT}(n,k) \geq \frac{1}{n}2^{- n/e^{k/e^2}}e^{-9n/2^k}$. So $1/p_{SAT}(n,k) \leq n2^{n/e^{k/e^2}}e^{9n/2^k}$. 
Thus, the probability that $|S| > 4 n^3 2^n/(\binom{n}{\alpha n} k^{\alpha n})$ is at most $2 \cdot 2^{-n/(32\sqrt{k})} \cdot n2^{n/e^{k/e^2}}e^{9n/2^k}$, if $m$ in is in the range $(d_k - 1)n \le m \le (d_k + 1)m$. Note that 
$2 \cdot 2^{-n/(32\sqrt{k})} \cdot n2^{n/e^{k/e^2}} \cdot e^{9n/2^k} \leq 2 \cdot 2^{-n/(64\sqrt{k})} .$

Furthermore, $m$ is in the appropriate range with probability at least $ 1- 2^{-n/(3\ln(2)2^k)}$, by Lemma~\ref{lem:boundingM}.
 
Thus, the probability that $S > 4 n^3 2^n/\left(\binom{n}{\alpha n} k^{\alpha n}\right)$ is at most $2\cdot 2^{-n/(3\ln(2)2^k)}$ conditioned on the formula being satisfiable.
\end{proof}

\section{Bounding the True Positive Rate}
\label{sec:falseNegative}

To demonstrate the correctness of our algorithm, we must show that if the input formula $\phi$ is satisfiable we will find some satisfying assignment with high probability. For this we must argue that a substantial fraction of randomly chosen assignments with low Hamming distance to a satisfying assignment will satisfy a large number of clauses. 

First, we will define a \emph{bad formula} as a formula with too few true positives, i.e. when too few of the small-Hamming-distance assignments satisfy a large number of clauses. We will then bound how often a formula in the \emph{planted} distribution, $\DPa$, is bad\footnote{We will achieve an upper bound on how often a formula is bad by bounding how often a formula has no \emph{single} satisfying assignment which has the desired number of true positives in its Hamming ball of radius $\alpha n$. This bound is, of course, not tight. However, this does give an upper bound on how often a formula can have too few false positives.}. We will show that if an event (such as a formula being bad) happens with low enough probability in $\DPa$, then it necessarily also has low probability in the real distribution, $D_R$. Finally we will combine these results to get the desired result that \HammDist~runs correctly with high probability.

Throughout this section, ``small Hamming distance'' refers to a Hamming distance less than or equal to $\alpha n$. 
For our algorithm to run as efficiently as possible, $\alpha$ should be chosen to be $\Theta\left(\frac{\lg(k)}{k}\right)$. Throughout this section, we will use $T$ to refer to the threshold of our algorithm: the number of satisfied clauses above which we consider it worthwhile to run a local search for a satisfying assignment. We further explore assignments that satisfy at least $T$ clauses. Our algorithm uses the following threshold: $T =\left(1-\frac{1-(1-\alpha)^{2k}}{2^k-1}\right)m$.

The end goal of this section is to prove that conditioned upon a formula being satisfiable, one of its satisfying assignments has an $(\alpha n)$-Hamming ball with at least $\binom{n}{\alpha n}\frac{1}{2}$ assignments that are above the threshold (and so are true positives). The \HammDist~algorithm will sample roughly $n^2$ assignments in the $(\alpha n)$-Hamming ball of that satisfying assignment. Thus, with high probability, one such small-Hamming-distance assignment will be randomly sampled by our algorithm. Therefore, we find a satisfying assignment, if one exists, with high probability.

\paragraph{Bounding a Bad Event in the Planted Distribution}

We will first define $S_{bad}$, which will contain all formulas for which the fraction of false negatives is too high. $S_{bad}$ will be a superset of all formulas that have too few false positives. $S_{bad}$ contains some formulas on which our algorithm will run correctly; however, we can use it to get an upper bound on how often our algorithm will fail, because all formulas on which our algorithm won't succeed on with high probability lie in $S_{bad}$.

\begin{definition}
	Let $H(\vec{a},\alpha, n)$ be the set of all assignments in $\{0,1\}^n$ with Hamming distance at most $\alpha n$ from $\vec{a}$.	
\end{definition}

\begin{definition}
    Let $H_=(\vec{a},\alpha, n)$ be the set of all assignments in $\{0,1\}^n$ with Hamming distance exactly $\alpha n$ from $\vec{a}$.	
\end{definition}

\begin{definition}
	Let $S_{bad}(m,n,k,\alpha, T)$ be the set of all formulas $\phi$ for which:
	\begin{itemize}
		\item $\phi$ has at least one satisfying assignment
		\item Every assignment $\vec{a}$ which satisfies $\phi$ has the following property: fewer than $1/2$ of the assignments $\vec{h} \in H_=(\vec{a},\alpha, n)$ satisfy at least $m-T$ of the clauses of $\phi$.
        \item Have $m$ clauses of $k$ literals each, and at most $n$ variables total.  
	\end{itemize}      
\end{definition}

In other words, $S_{bad}$ is the set of satisfiable formulas where most of the small-Hamming-distance assignments don't satisfy a large number of clauses. Formulas in this set pose a problem for our algorithm, because we can't reliably identify assignments that are a small Hamming distance from their satisfying assignment. In this section, we'll show that this set of formulas makes up an exponentially small fraction of the formulas in $U_\phi$.

In the lemma below, we express an expectation over small-Hamming-distance assignments in terms of an expectation over random assignments and an expectation over assignments that are a small Hamming distance from a \emph{falsifying} assignment---an assignment that makes every literal in a given clause false. This is helpful because while the former quantity is what we care about, the latter two qualities are more easily analyzed.

\begin{lemma}
Let $\vec{a} \in \{0,1\}^n$ be an assignment of length $n$. Let $c$ be a clause of size $k$. 
Let $r_c(\vec{a})$ be the number of literals in $c$ satisfied by some assignment $\vec{a}$. 


Let $\mathcal{H}$ be the uniform distribution over assignments in $H_=(\vec{a}, \alpha, n)$.

Let $\{f'(\cdot) : [0, k] \rightarrow \mathbb{R}\}$ be any function taking as input the number of literals satisfied in $c$ and returning a real number. Let $f(c,\vec{a}) = f'(r_c(\vec{a}))$.

$\DPc(n, k, \vec{a})$ is the uniform distribution over clauses that are satisfied by $\vec{a}$.

$D_{fc}(n, k, \vec{a})$ is the uniform distribution over clauses that are \emph{not} satisfied by $\vec{a}$.

$D_{replace}(k)$ is the uniform distribution over \emph{all} clauses. 

    

 Then, we have:
 \begingroup
 \Large
 \begin{equation*}
    E_{c \sim \DPc(n, k, \vec{a}), \vec{h} \sim \mathcal{H}}[f(c,\vec{h}))] = 
\end{equation*}
\endgroup
 \vspace{-2.5ex}
 \begingroup
 \Large
 \begin{equation}
  \frac{2^k}{2^k - 1} E_{c \sim D_{replace}(n,k), \vec{h} \sim \mathcal{H}}[f(c,\vec{h})]- \frac{1}{2^k-1} E_{c \sim D_{fc}(n, k, \vec{a}), \vec{h} \sim \mathcal{H}}[f(c,\vec{h})].
\end{equation}
\endgroup
        
    

\label{lem:generalSmallHDf}
\end{lemma}
\begin{proof}
Let $S(\cdot)$ be the support of the distribution $(\cdot)$.
Note that $S(\DPc(n, k, \vec{a})) \cap S( D_{fc}(n, k, \vec{a})) = \emptyset$. Further note that $ S(D_{replace}(n,k)) =S(\DPc(n, k, \vec{a})) \cup S(D_{fc}(n, k, \vec{a}))$. So, these distributions are non-overlapping and all uniform. Moreover, $|S(\DPc(n, k, \vec{a}))| = \frac{2^k-1}{2^k}|S(D_{replace}(n, k))|$, and $|S(D_{fc}(n, k, \vec{a}))| = \frac{1}{2^k-1}|S(D_{replace}(n, k))|$.

As a result we can say the following:

\begin{align}
E_{c \sim D_{replace}(n,k), \vec{h} \sim \mathcal{H}}[f(c,\vec{h})] &= \frac{2^k-1}{2^k} E_{c \sim \DPc(n, k, \vec{a}), \vec{h} \sim \mathcal{H}}[f(c,\vec{h}))]+ \frac{1}{2^k} E_{c \sim D_{fc}(n, k, \vec{a}), \vec{h} \sim \mathcal{H}}[f(c,\vec{h})] \\
\frac{2^k-1}{2^k} E_{c \sim \DPc(n, k, \vec{a}), \vec{h} \sim \mathcal{H}}[f(c,\vec{h}))] &=E_{c \sim D_{replace}(n,k), \vec{h} \sim \mathcal{H}}[f(c,\vec{h})] - \frac{1}{2^k} E_{c \sim D_{fc}(n, k, \vec{a}), \vec{h} \sim \mathcal{H}}[f(c,\vec{h})] \\
 E_{c \sim \DPc(n, k, \vec{a}), \vec{h} \sim \mathcal{H}}[f(c,\vec{h}))] &= \frac{2^k}{2^k-1} \left(E_{c \sim D_{replace}(n,k), \vec{h} \sim \mathcal{H}}[f(c,\vec{h})] - \frac{1}{2^k} E_{c \sim D_{fc}(n, k, \vec{a}), \vec{h} \sim \mathcal{H}}[f(c,\vec{h})] \right)\\
  E_{c \sim \DPc(n, k, \vec{a}), \vec{h} \sim \mathcal{H}}[f(c,\vec{h}))] &= \frac{2^k}{2^k-1} E_{c \sim D_{replace}(n,k), \vec{h} \sim \mathcal{H}}[f(c,\vec{h})] - \frac{1}{2^k-1} E_{c \sim D_{fc}(n, k, \vec{a}), \vec{h} \sim \mathcal{H}}[f(c,\vec{h})] 
\end{align}
\end{proof}

\begin{corollary}
Given $c \sim \DPc(n,k,\vec{a})$ and $\vec{h} \sim \mathcal{H}$, then $\vec{h}$ leaves $c$ \emph{unsatisfied} with probability $\left(1 - \left(1-\alpha\right)^k \right)/\left(2^k-1 \right)$. 
\label{cor:boundSmallHDCl}
\end{corollary}
\begin{proof}
We will use  Lemma \ref{lem:generalSmallHDf} and the function $f'(x)= \begin{cases}
0 \text{ if } x>0\\
1 \text{ if } x=0
\end{cases}$. In other words, the function is $1$ when the clause is unsatisfied and $0$ when it is satisfied.

For this case, $E_{c \sim D_{replace}(n,k)}[f(c,\lambda(\vec{a}))] = \frac{1}{2^k}$ because given any assignment and a random clause, the probability the assignment falsifies the clause is $\frac{1}{2^k}$.

Note that in the above, the $k$ literals are indeed independently falsified with probability $\frac{1}{2}$. This is because the clause $c$ is chosen such that each literal in the clause is picked independently and uniformly at random from the $2n$ possible literals (by the definition of $D_{replace}$). 

Additionally $E_{c \sim D_{fc}(n, k, \vec{a}), \vec{h} \sim \mathcal{H}}[f(c,\vec{h})] = (1-\alpha)^k$, because if all literals start out false, then the clause is only remains falsified if none of the variables in the clause are chosen from the set of $\alpha n$ variables that have their values flipped. 

Once again, the $k$ literals are independently falsified, this time with probability $1-\alpha$. This is because the clause $c$ is chosen such that each literal in the clause is picked independently and uniformly at random from the $n$ possible literals that are falsified by $\vec{a}$.
    
Using Lemma \ref{lem:generalSmallHDf} we have that

\begin{align}
E_{c \sim \DPc(n, k, \vec{a}), \vec{h} \sim \mathcal{H}}[f(c,\vec{h}))] &= \frac{2^k}{2^k - 1} E_{c \sim D_{replace}(n,k)}[f(c,\lambda(\vec{a}))]- \frac{1}{2^k-1} E_{c \sim D_{fc}(n, k, \vec{a}), \vec{h} \sim \mathcal{H}}[f(c,\vec{h})]\\
E_{c \sim \DPc(n, k, \vec{a}), \vec{h} \sim \mathcal{H}}[f(c,\vec{h}))] &= \frac{2^k}{2^k - 1} \frac{1}{2^k}- \frac{1}{2^k-1} (1-\alpha)^k\\
E_{c \sim \DPc(n, k, \vec{a}), \vec{h} \sim \mathcal{H}}[f(c,\vec{h}))] &= \frac{1}{2^k - 1} - \frac{1}{2^k-1} (1-\alpha)^k\\
E_{c \sim \DPc(n, k, \vec{a}), \vec{h} \sim \mathcal{H}}[f(c,\vec{h}))] &= \frac{1-(1-\alpha)^k}{2^k - 1}.
\end{align}

Thus, the probability that clause chosen at random from $\DPc(n, k, \vec{a})$ is \emph{unsatisfied} by $\vec{h} \sim \mathcal{H}$ is $\frac{1-(1-\alpha)^k}{2^k - 1}$.
\end{proof}

\begin{lemma}

Let $m = (\ln(2) 2^k - x)n$	where $x<\ln(2) 2^k /2$.
	
Let $\alpha < 1/2$ and $k > 2$.

Choose an assignment, $\vec{a}$, at random from $U_{\vec{a}}$ and fix it. Now pick a vector $\vec{h}$ uniformly at random from among $H(\vec{a}, \alpha, n)$. Finally pick a formula, $\phi$, uniformly at random from the planted distribution, $\DPa(m, n, k, a)$.

Let $\ncus(\phi, \vec{h})$ be the number of \textit{unsatisfied} clauses when the variables in the formula $\phi$ are set according to the assignment $\vec{h}$.

Let $\pbad$ be $Pr[\ncus{\phi, \vec{h}}> m\left(1-(1-\alpha)^{2k}\right)/\left(2^k-1\right)]$.

Then 
$$\pbad<e^{-(1-\alpha)^{2k}\ln(2)n/(16)}$$

\label{lem:boundExistenceOfOutliers}
\end{lemma}
\begin{proof}



Let $p_f$ be the probability that a given clause is falsified by $\vec{h}$.

Let $p_{HD}(d)$ be the probability a given clause is falsified by a uniformly random assignment with Hamming distance exactly $d$ from $\vec{a}$.

If $\HD(\vec{h}, \vec{a}) = \alpha n$, then by Corollary \ref{cor:boundSmallHDCl} 
 we have that $p_{HD} (\alpha n)= (1-(1-\alpha)^{k})/(2^k-1)$. If instead $\HD(\vec{h}, \vec{a}) = \alpha' n$, then by Corollary~\ref{cor:boundSmallHDCl} 
  we have that $p_{HD} (\alpha' n)= (1-(1-\alpha')^{k})/(2^k-1)$. If $\alpha'\leq \alpha$ then $p_{HD} (\alpha' n)\leq  p_{HD} (\alpha n)$.
Thus, $p_f\leq (1-(1-\alpha)^{k})/(2^k-1)$.

So the mean number of clauses falsified is $\mu = p_f m$.

Let $p_{new} = (1-(1-\alpha)^{2k})/(2^k-1)$. Note that $p_{new} > p_f$.
In our algorithm, we chose the threshold to be $T=  p_{new}m$. 
So 
$$T/\mu = p_{new}/p_f = \frac{1-(1-\alpha)^{2k}}{1-(1-\alpha)^{k}} = 1+(1-\alpha)^k.$$

Now let $\delta_\mu = (1-\alpha)^k<1$. Using the multiplicative Chernoff bound: 

\begin{align}
Pr[\ncus(\phi, \vec{h})]> T = (1+\delta_\mu)\mu]&<e^{-\delta_\mu^2\mu/3} \\
&<e^{-(1-\alpha)^{2k}p_fm/3}\\
&<e^{-(1-\alpha)^{2k}\ln(2)n/(12)}
\end{align}

\end{proof}

\begin{lemma}
Let $m = (\ln(2) 2^k - x)n$	where $x<\ln(2) 2^k /2$.
Let $\alpha < 1/2$ and $k >2$.

Choose an assignment, $\vec{a}$, uniformly at random and fix it. Now pick a formula, $\phi$, uniformly at random from the planted distribution about $\vec{a}$, $\DPa(m, n, k, \vec{a})$. Let $\tilde{T}$ be 

$$m \frac{1-(1-\alpha)^{2 k}}{2^k-1}.$$

Let $S_{\alpha n}(\phi)$ be the set of all assignments in $H(\vec{a},\alpha, n)$ that leave more than $\tilde{T}$ clauses in $\phi$ unsatisfied. Then,
	
$$Pr_{\phi}[|S_{\alpha n}(\phi)| > \frac{1}{2}|H(\vec{a},\alpha,n)|] \leq 2 e^{-(1-\alpha)^{2k}\ln(2)n/(16)}$$
\label{cor:badFormulasBoundedInPlant}
\end{lemma}
\begin{proof}
Let $p_{FalseNeg}$ be the probability that a random small-Hamming-distance assignment in $H(\vec{a},\alpha,n)$ has more than 
$\tilde{T}$ unsatisfied clauses. By Lemma~\ref{lem:boundExistenceOfOutliers}, $p_t\leq e^{-(1-\alpha)^{2k}\ln(2)n/(16)}$. 


Let $\Sbad$ be the set of formulas for which 
$$|S_{\alpha n}(\phi)| > \frac{1}{2}|H(\vec{a},\alpha,n)|.$$ Put another way, let $\Sbad$ be the set of formulas with at least half of their small-Hamming-distance assignments leaving more than $\tilde{T}$ clauses unsatisfied. Then, let $\pbad$ be the fraction of formulas in $\Sbad$.

Then, 

$$p_{FalseNeg} = \frac{1}{|D_{pa}(m,n,k,\vec{a})|} \sum_\phi \frac{|S_{\alpha n}(\phi)|}{|H(\vec{a},\alpha,n)|} \ge \frac{1}{2}\pbad.$$

In the above, we use $|D_{pa}(m,n,k,\vec{a})|$ to mean the size of the support of $D_{pa}$. For a visual representation of the relationship between $p_{FalseNeg}$, $\Sbad$, and $\pbad$, see Figure \ref{fig:volumeSalphan}.

 Then $p_{FalseNeg} \geq \pbad/2$, so $\pbad \leq 2p_{FalseNeg} \leq 2e^{-(1-\alpha)^{2k}\ln(2)n/(16)}$.
\end{proof}

\begin{figure}
	\centering
	\includegraphics[width=4in]{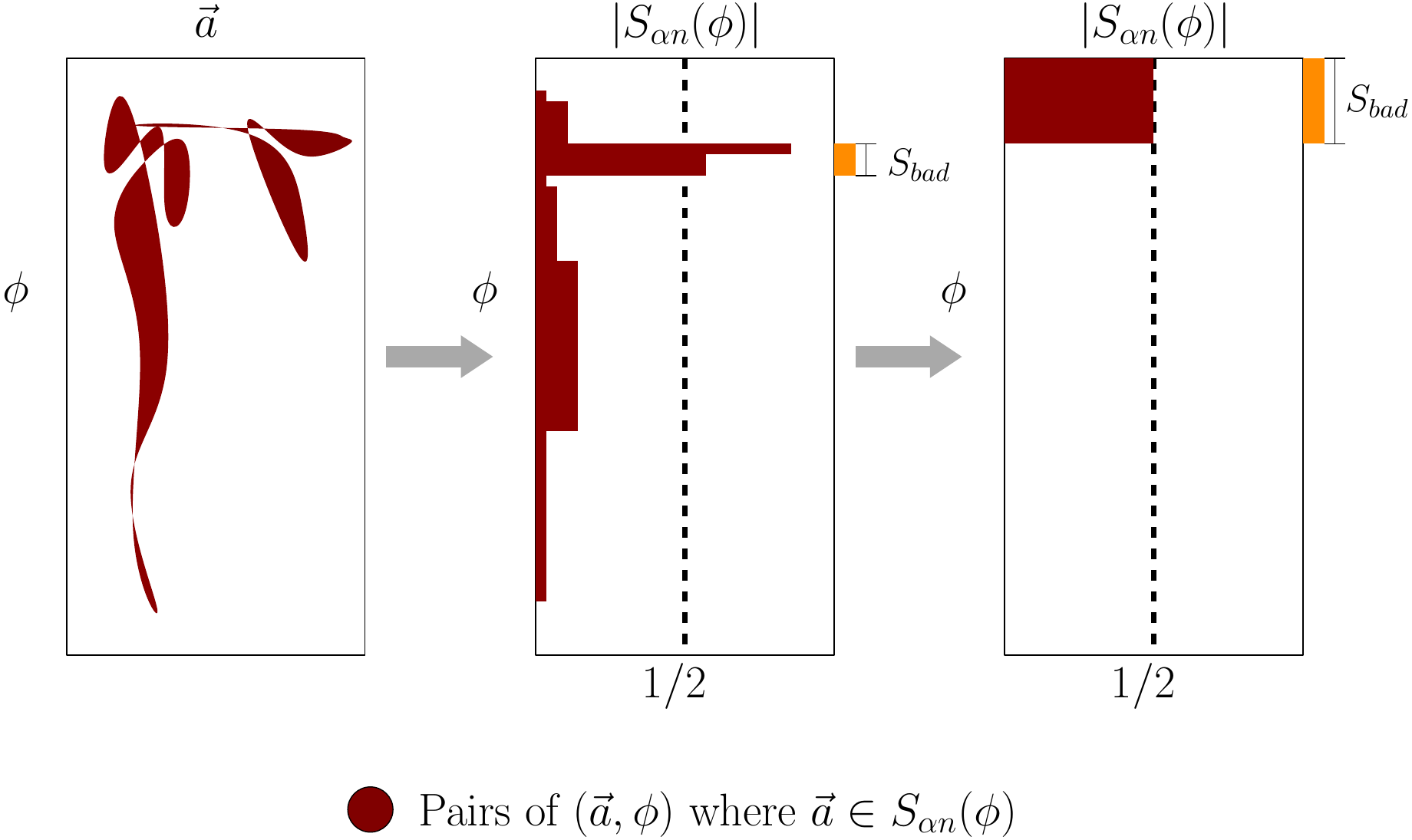}
	\caption{Three representations of $(\phi,\alpha)$ pairs, where pairs such that $\vec{a} \in S_{\alpha n}(\phi)$ are dark red and other pairs are white. Across the three representations, the total dark red area is meant to be fixed. \\
	In the first representation we mark in red every pair $(\phi,\alpha)$ where $\vec{a} \in S_{\alpha n}(\phi)$. The vertical axis represents different formulas, while the horizontal axis represents different assignments. \\
	In the second representation we instead mark the size of $|S_{\alpha n}(\phi)|$ for every $\phi$, starting from the left. The orange bar highlights $\phi \in \Sbad$, i.e. those $\phi$ for which $|S_{\alpha n}(\phi)| > \frac{1}{2}|H(\vec{a},\alpha,n)|$.\\
	Finally, we show a visualization of the worst-case distribution of $|S_{\alpha n}(\phi)|$, given a fixed dark red area. The orange bar highlights which $\phi$ would be in $\Sbad$ in this case. While the length of the orange bar in the middle figure represents the true value of $\pbad$, the length of the orange bar in the rightmost figure represents our upper bound on $\pbad$ as proved in Lemma~\ref{cor:badFormulasBoundedInPlant}.
    }
	\label{fig:volumeSalphan}
\end{figure}





\paragraph{Bounds in the Planted Distribution Imply Bounds in Uniform Distribution}

We will now show that algorithms that work with high enough probability in the planted distribution $\DPa(m, n, k)$ work with high probability given random formulas drawn from $D_R(m, n, k)$. Similar reductions have been shown in previous work \cite{plantedToRandom1,plantedToRandom2,vyas2018super}.

\begin{definition}
Let the distribution $\DPa(m, n, k)$, the planted distribution, be the distribution formed by uniformly selecting from all pairs $(\phi,\vec{a})$, where $\vec{a}$ is a satisfying assignment to $\phi$ and $\phi$ has $m$ clauses.
\end{definition}

\begin{corollary}
Samples from the distribution $\DPa(m,n,k)$ can be generated by first uniformly picking an assignment $\vec{a} \in \{0,1\}^n$, and then picking $m$ clauses uniformly from the set of all $k$-length clauses that $\vec{a}$ satisfies.  
\end{corollary}
\begin{proof}
This will uniformly generate all pairs of $\vec{a}$ and $\phi$ where $\vec{a}$ satisfies $\phi$. 
\end{proof}

\begin{lemma}
The number of formulas of length $m$ (or, equivalently, the size of the support of $D_R(m, n, k)$) is $c_{r}(m)=(n^k \cdot 2^k)^m$.

The number of pairs $(\phi , \vec{a})$ where $\phi$ is a formula of length $m$ and $\vec{a}$ is a satisfying assignment (or, equivalently, the size of the support of the planted distribution) is $c_{p}(m)= 2^n(n^k \cdot (2^k-1))^m$. 

The ratio $c_{p}(m)/c_{r}(m) \leq 2^n \cdot e^{-m/2^k}$.

\label{lem:bassic numbers}
\end{lemma}
\begin{proof}

The number of formulas of length $m$ is $c_{r(m)}=((2n)^k)^m$ because for each of the $m$ clauses, there are $k$ choices to be made from $2n$ possible literals.
	 
The support of the planted distribution has size $2^n(n^k \cdot (2^k-1))^m$ because there are $2^n$ possible choices of assignment, and for each of those, there are $(n^k \cdot (2^k-1))^m$ possible choices of formula that are satisfied by that assignment.

We can demonstrate the lemma's final statement as follows:
\begin{align}
c_{p}(m)/c_{r}(m) &= 2^n(n^k \cdot (2^k-1))^m/(n^k \cdot 2^k)^m\\
c_{p}(m)/c_{r}(m) &= 2^n((2^k-1)/2^k)^m\\
c_{p}(m)/c_{r}(m) &= 2^n \cdot (1-1/2^k)^{m}\\
c_{p}(m)/c_{r}(m) &\leq 2^n \cdot e^{-m/2^k}
\end{align}
\end{proof}

\newcommand{\Sg}{S_{\gamma}^{m,n,k}}

The following lemmas allow us to lower-bound the probability of a formula being bad when drawn from $D_R$, given a bound on the probability of a formula being bad when drawn from $D_P$. 

\begin{lemma}	
Let $\Sg$ be an arbitrary subset of satisfiable formulas with exactly $m$ clauses, at most $n$ variables, and $k$ literals per clause.	

Let $p_{sat}$ be the probability that a $\phi$ drawn from $D_R(m,n,k)$ is satisfiable. 	
	
If the probability that $\phi$ is in $\Sg$ is less than or equal to $p_{\gamma}$ when $(\phi,\vec{a})$ is drawn from $\DPa(m,n,k)$, then the probability that  $\phi$ is in $\Sg$ is at most $p_{\gamma}\cdot (2^n \cdot e^{-m/2^k})$ when $\phi$ is drawn from $D_R(m,n,k)$.
\label{lem:plantedToAllReal}
\end{lemma}
\begin{proof}
By Lemma \ref{lem:bassic numbers} we have that the support of $\DPa(m,n,k)$ has size $c_p =2^n (n^k(2^k-1))^m$. Furthermore, $\DPa(m,n,k)$ draws uniformly over this support. 
	
By Lemma \ref{lem:bassic numbers} we have that the support of $D_R(m,n,k)$ has size $c_r =(n^k2^k)^m$. Furthermore, $D_R(m,n,k)$ draws uniformly over this support. 

Because $D_R(m,n,k)$ draws uniformly over formulas $\phi$ in its support, the probability that a $\phi$ drawn from $D_R(m,n,k)$ is in $\Sg$ is equal to $|\Sg|/c_r$.

Draw a tuple $(\phi, \vec{a})$ from $\DPa(m,n,k)$. Given our condition on $p_\gamma$ in the lemma, the total number of tuples $(\phi, \vec{a})$ in the support of $\DPa(m, n, k)$ where $\phi \in \Sg$ is less than or equal to $c_p p_\gamma.$
	
$|\Sg| \le c_p p_\gamma$, because every $\phi \in \Sg$ appears in at least one tuple in $\DPa(m,n,k)$ (and possibly in several). We know this because $|\Sg|$ can contain only satisfied formulas, and every satisfied formula must appear at least once in a tuple in $\DPa(m,n,k)$. For a visual representation of the connection between $\DPa(m,n,k)$ and $D_R(m,n,k)$, see Figure \ref{fig:sgammaCPCR}.

    



Therefore, the probability that $\phi$ drawn from $D_R(m,n,k)$ are in $\Sg$ is $|\Sg|/c_r \le p_{\gamma} \cdot c_p/c_r$. From Lemma~\ref{lem:bassic numbers}, we have that $|\Sg| \leq p_{\gamma}\cdot (2^n \cdot e^{-m/2^k})$.
\end{proof}

\begin{figure}
	\centering
	\includegraphics[width=4in]{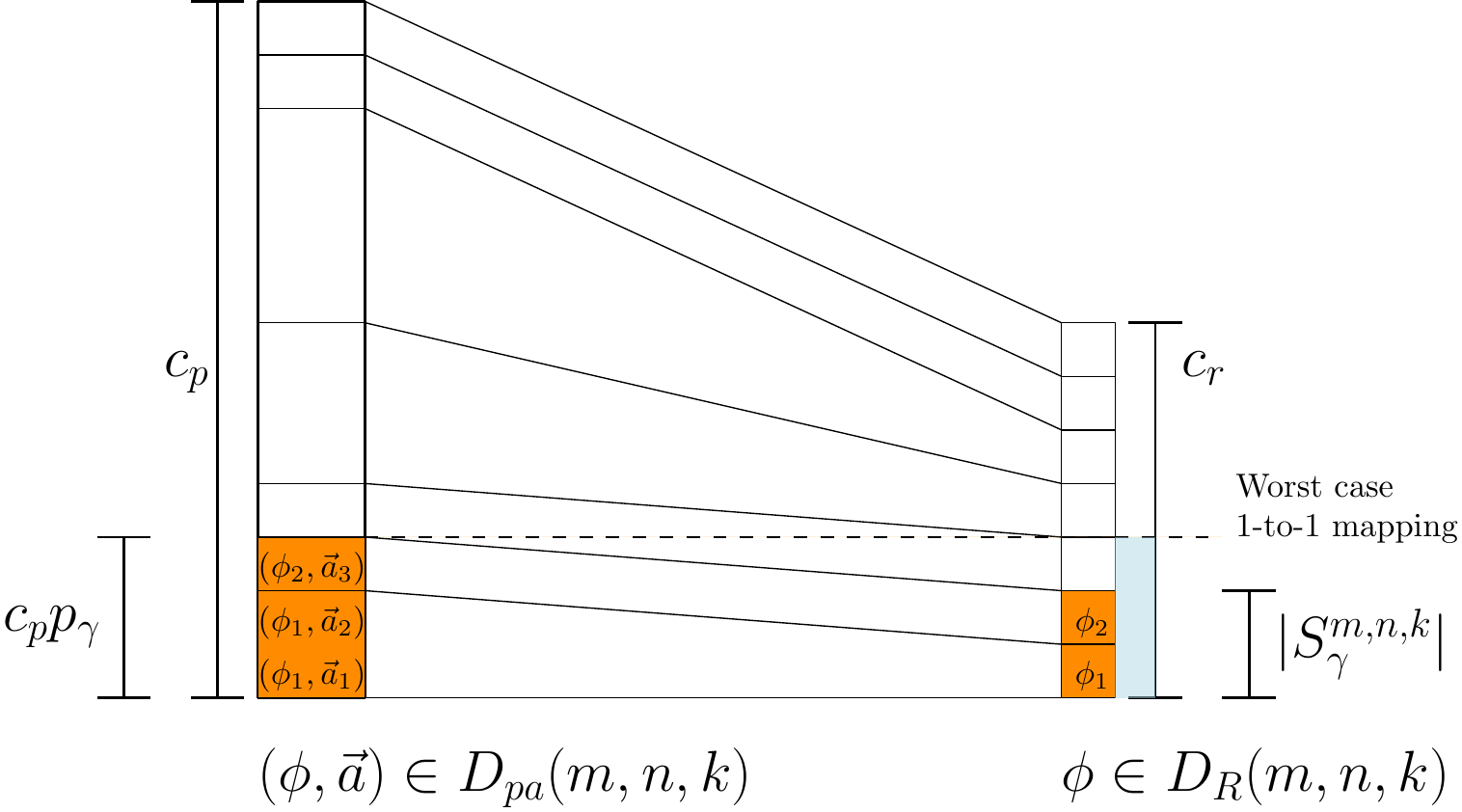}
	\caption{The left side of this diagram represents all pairs of $(\phi, \vec{a})$ in the support of $\DPa(m,n,k)$. The right side of this diagram represents all $\phi$ in the support of $D_R(m,n,k)$. The orange in the diagram represents pairs $(\phi,\vec{a})$ where $\phi \in \Sg$ (on the left) and formulas $\phi$ in the set $\Sg$ (on the right). The diagram shows an example of a possible mapping. The dotted line shows the worst-case 1-to-1 mapping that we assume to get our upper bound on $|\Sg|$.}
	\label{fig:sgammaCPCR}
\end{figure}

\begin{corollary}
	Let $p_{sat}$ be the probability that a $\phi$ drawn from $D_R(m,n,k)$ is satisfiable.
    
	Let $\Sg$ be a set of formulas which only contains formulas with at least one satisfying assignment. 
	
	If the probability that $\phi$ is in $\Sg$ is at most $p_{\gamma}$ when $(\phi,\vec{a})$ is drawn from $\DPa(m,n,k)$, then  the probability that $\phi$ is in $\Sg$ is at most $p_{\gamma}\cdot (2^n \cdot e^{-m/2^k}) \cdot \frac{1}{p_{sat}}$ when $\phi$ is drawn from the satisfied formulas of $D_S(m,n,k)$.
	\label{lem:plantedToSatReal}
\end{corollary}
\begin{proof}
Using Lemma \ref{lem:plantedToAllReal}, we have that the probability of $\phi \in \Sg$ over $D_R(m,n,k)$ is $\leq p_{\gamma}\cdot (2^n \cdot e^{-m/2^k})$. 

The support of $D_R(m,n,k)$ is $1/p_{sat}$ times larger than the support of $D_S(m,n,k)$, and both are uniform over their support. Furthermore, the support of $D_S(m,n,k)$ is a subset of the support of $D_R(m,n,k)$. The prevalence of $\phi \in \Sg$ can therefore only increase by a factor of $1/p_{sat}$.
\end{proof}


\begin{corollary}
	If $m\geq  (2^k \ln(2) - c)n$ for some constant $c$, then the following holds:
	
	Let $p_{sat}$ be the probability that a $\phi$ drawn from $D_R(m,n,k)$ is satisfiable. 	
	Let $\Sg$ be a set of formulas which only contains formulas with at least one satisfying assignment. 
	
	If the probability that $\phi$ is in $\Sg$ is at most $p_{\gamma}$ when $(\phi,\vec{a})$ is drawn from $\DPa(m,n,k)$, then  the probability that  $\phi$ is in $\Sg$ is at most $p_{\gamma}\cdot (e^{cn/2^k}) \cdot \frac{1}{p_{sat}}$ when $\phi$ is drawn from $D_S(m,n,k)$.
	\label{lem:plantedToSatRealSetM}
\end{corollary}
\begin{proof}
From Corollary \ref{lem:plantedToSatReal}, we have that the probability $q$ that $\phi$ is in $\Sg$ when $\phi$ is drawn from $D_S(m,n,k)$ is at most $p_{\gamma}\cdot (2^n \cdot e^{-m/2^k}) \cdot \frac{1}{p_{sat}}$.

	\begin{align}
	q& \leq p_{\gamma}\cdot (2^n \cdot e^{-m/2^k}) \cdot \frac{1}{p_{sat}}\\
	&\leq p_{\gamma}\cdot (2^n \cdot e^{-((2^k \ln(2) - c)n)/2^k}) \cdot \frac{1}{p_{sat}}\\
	&\leq p_{\gamma}\cdot (2^n \cdot 2^{-n} \cdot e^{cn/2^k}) \cdot \frac{1}{p_{sat}}\\
	&\leq p_{\gamma}\cdot e^{cn/2^k} \cdot \frac{1}{p_{sat}}
	\end{align}
\end{proof}

\paragraph{Putting Everything Together to get Algorithm Correctness}

\begin{corollary}
Let $k\geq \kstr$.\\	
Let $p_{sat}$ be the probability that a $\phi$ drawn from $D_R(m,n,k)$ is satisfiable. 

Let $m=(\ln(2)2^k-c)n$. Let $T = (1-\frac{1-(1-\alpha)^{2k}}{2^k-1})m)$.
	
If the probability that $\phi$ is in 
$S_{bad}(m,n,k,\alpha,T)$ 
is at most $p_{bad}$ when $(\phi,\vec{a})$ is drawn from $\DPa(m,n,k)$, then the probability that $\phi$ is in $S_{bad}(m,n,k,\alpha,T)$ is at most $(p_{bad}\cdot e^{cn/2^k} \cdot \frac{1}{p_{sat}})$ when $\phi$ is drawn from the satisfied formulas of $D_S(m,n,k)$.
\label{cor:helpfulPlantedTo}
\end{corollary}
\begin{proof}
This comes from applying Corollary \ref{lem:plantedToSatRealSetM}.
Every $\phi$ that in $S_{bad}$ in the planted distribution can map to \emph{at most} one formula in $S_{bad}$ in the real distribution. 
\end{proof}

\begin{lemma}
Let $k\geq \kstr$. Let $m=(\ln(2)2^k-c)n$. \\
Let $p_{sat}$ be the probability that a formula drawn from $D_R(m, n, k)$ is satisfiable. 		
	
Let $\phi$ be a formula drawn from $D_S(m,n,k)$. Let $S$ be a sequence of $2n^2 2^n/\binom{n}{\alpha n}$ assignments, each drawn independently and uniformly at random from $U_{\vec{a}}(n)$. 

Then, with probability at least $1-e^{cn/2^k}2 e^{-(1-\alpha)^{2k}\ln(2)n/(16)}\frac{1}{p_{sat}}-e^{-n^2}$, there is an assignment $\vec{h}\in S$ that simultaneously 
\begin{enumerate}[(a)]
	\item has Hamming distance at most $\alpha n$ from a satisfying assignment to $\phi$ and
	\item satisfies at least $(1-\frac{1-(1-\alpha)^{2k}}{2^k-1})m$ clauses in $\phi$. 
\end{enumerate}
\label{lem:thisIsALabel}
\end{lemma}
\begin{proof}
First we apply both Corollary \ref{cor:badFormulasBoundedInPlant} and Corollary \ref{cor:helpfulPlantedTo} to show that with probability at least $1-e^{cn/2^k}2 e^{-(1-\alpha)^{2k}\ln(2)n/(16)}\frac{1}{p_{sat}}$, the set $X$ of assignments $\vec{h}$ that satisfy both (a) and (b) is of size $|X|\geq \frac{1}{2}|H(\vec{a},\alpha,n)|$.

Now, when we are sampling uniformly at random from this set, the probability that we draw no elements from $X$ is at most:
$$\left(1-\frac{1}{2}|H(\vec{a},\alpha,n)|\cdot \frac{1}{2^n}\right)^{2n^2 2^n/\binom{n}{\alpha n}} \leq \left(1-\frac{1}{2}\binom{n}{\alpha n}\cdot \frac{1}{2^n}\right)^{2n^2 2^n/\binom{n}{\alpha n}} \leq e^{-n^2}.$$

So the probability that a vector $\vec{h}$ that satisfies both (a) and (b) is drawn is at least $1-e^{cn/2^k}2 e^{-(1-\alpha)^{2k}\ln(2)n/(16)}\frac{1}{p_{sat}}-e^{-n^2}$ by the union bound.
\end{proof}

Now, for any given \emph{fixed} choice of $m$ we have the desired result, which is that the false negative rate is very low. However, recall that by the definition of the threshold used by Ding, Sly and Sun~\cite{threasholdSAT}, the value of $m$ is chosen randomly from a Poisson distribution. We will now show that with high probability, $m$ takes a value such that our algorithms still has a low false negative rate.

\begin{lemma}
Let $k\geq \kstr$.\\
Let $\phi$ be a formula drawn uniformly at random from $D_{\Phi}(n,k)$. Let $S$ be a sequence of $2n^2 2^n/\binom{n}{\alpha n}$ assignments drawn independently and uniformly at random from $U_{\vec{a}}(n)$. 

Let $1/p_{sat}^{\star} = \max_m\left(\frac{1}{p_{sat,m}}\right)$ be the maximum value of $1/p_{sat,m}$ for $m\in [(\delta(k)-1)n,(\delta(k)+1)n]$, where $p_{sat,m}$ is the probability a formula is satisfiable when the formula is drawn from $D_{R}(n,m,k)$.

Then with probability at least
$$1-e^{-n/(3\ln(2)2^k)}-\frac{e^{3n/2^k}2 e^{-(1-\alpha)^{2k}\ln(2)n/(16)}}{p_{sat}^\star}-e^{-n^2},$$
 there is an assignment $\vec{h}\in S$ that simultaneously 
\begin{enumerate}[(a)]
	\item has Hamming distance at most $\alpha n$ from a satisfying assignment to $\phi$ and
	\item satisfies at least $(1-\frac{1-(1-\alpha)^{2k}}{2^k-1})m$ clauses in $\phi$. 
\end{enumerate}
\label{lem:phiIsSolvedIfSatProbOkay}
\end{lemma}
\begin{proof}
When $\phi$ is drawn from $D_{\Phi}(n,k)$, then by Lemma~\ref{lem:boundingM}, with probability at least $1-2 \cdot 2^{-n/(3\ln(2)2^k)}$, the number of clauses chosen is $m\in[(\delta(k)-1)n,(\delta(k)+1)n]$. 

Given that $\phi$ has $m\in[(\delta(k)-1)n,(\delta(k)+1)n]$ clauses, then by Lemma~\ref{lem:thisIsALabel}, the probability that $S$ has a $\vec{h}\in S$ that satisfies both (a) and (b) is at least $1-e^{3n/2^k}2 e^{-(1-\alpha)^{2k}\ln(2)n/(16)}\frac{1}{p_{sat,m}}-e^{-n^2}$, because $c \le 3$. 

So the probability that a $\phi$ drawn from $D_{\Phi}(n,k)$ has a  $\vec{h}\in S$ that satisfies both (a) and (b) is at least $1-e^{-n/(3\ln(2)2^k)}-e^{3n/2^k}2 e^{-(1-\alpha)^{2k}\ln(2)n/(16)}\frac{1}{p_{sat,m}}-e^{-n^2}$.
\end{proof}


Next, we want to bound the value of $p_{sat}^{\star}$.  

We will use Corollary \ref{cor:numSatIsHigh} from Section \ref{sec:lotsOfSAT} to lower bound the number of satisfying assignments.

\begin{reminder}{Corollary \ref{cor:numSatIsHigh}}
	Let $m = (\ln(2)2^k-c)n$.	
	
	The probability that $\phi$ drawn from $D_R(m,n,k)$  has at least one satisfying assignment is at least $$\frac{1}{n}2^{- n/e^{k/e^2}}e^{-3cn/2^k}.$$
	\label{cor:numSatIsHighreminderBody}
\end{reminder}
\begin{proof}
Proof given in Section \ref{sec:lotsOfSAT}.
\end{proof}

\begin{lemma}
	Let $k>\kstr$.
    
    Let $\alpha  \leq \frac{1}{\sqrt{k}}$.
	
	Let $\phi$ be a formula drawn uniformly at random from $D_{\Phi}(n,k)$. Let $S$ be a sequence of $2n^2 2^n/\binom{n}{\alpha n}$ assignments drawn independently and uniformly at random from $U_{\vec{a}}(n)$. 
	
	Then with probability 
	$$\geq 1-e^{-n/(3\ln(2)2^k)}-3 e^{-(1-\alpha)^{2k}\ln(2)n/(32)}$$ 
	there is an assignment $\vec{h}\in S$ that simultaneously 
	\begin{enumerate}[(a)]
		\item has Hamming distance at most $\alpha n$ from a satisfying assignment to $\phi$ and
		\item satisfies at least $(1-\frac{1-(1-\alpha)^{2k}}{2^k-1})m$ clauses in $\phi$. 
	\end{enumerate}
	\label{lem:algIsCorrect}
\end{lemma}
\begin{proof}
	Let $1/p_{sat}^{\star} = \max_m\left(\frac{1}{p_{sat,m}}\right)$ be the maximum value of $1/p_{sat,m}$ for $m\in [(\delta(k)-1)n,(\delta(k)+1)n]$, where $p_{sat,m}$ is the probability a formula is satisfiable when the formula is drawn from $D_{R}(n,m,k)$. Consider $m = (\ln(2)2^k-c)n$.
	
	From Corollary \ref{cor:numSatIsHigh} we have that 
	$$1/p_{sat,m} \leq n 2^{ne^{-k/e^2} }e^{3cn/2^k}.$$

Thus, 
	$$1/p_{sat}^{\star} \leq  n 2^{ne^{-k/e^2} }e^{9n/2^k},$$ since $c <3$.
		
	This yields a probability
	$$\geq 1-e^{-n/(3\ln(2)2^k)}-n e^{n/2^k}2 e^{-(1-\alpha)^{2k}\ln(2)n/(16)}2^{ne^{-k/e^2}}e^{9n/2^k} -e^{-n^2}$$ 
	
    that there exists an $\vec{h} \in S$ satisfying this lemma's conditions.
    
	Because of the bounds on $\alpha$ and $k$ in the lemma's statement, the probability bound above can be simplified to 
$$1-e^{-n/(3\ln(2)2^k)}-3 e^{-(1-\alpha)^{2k}\ln(2)n/(32)}.$$ 
\end{proof}


\section{Lower Bounding the Number of Satisfied Formulas}
\label{sec:lotsOfSAT}
We will now show that $\phi$ drawn from $D_R(m,n,k)$ (the uniform distribution over formulas of length $m$) has a probability of being satisfiable of at least $n2^{- ne^{-k/e^2}}e^{-9n/2^k}$, assuming $m\in [(d_k-1)n,(d_k+1)n]$.


\begin{lemma}
	Let $m = (\ln(2)2^k-c)n$.
	Let $f_{count}(\phi)$ be the number of satisfying assignments of $\phi$.
	
	Let
	$$y = E[f_{count}(\phi)]_{\phi \sim \DPa(m,n,k)}.$$
	Then, 
	the probability that $\phi$ drawn from $D_R(m,n,k)$  has at least one satisfying assignment is at least
	$$e^{-cn/2^k}/y.$$
	
	\label{lem:plantedToRealSatCount}
\end{lemma}
\begin{proof}
	Recall that both $D_R(m,n,k)$ and $\DPa(m,n,k)$	are uniform over their supports. \\
	Let $c_p = 2^n (n^k(2^k-1))^m$ be the size of the support of $\DPa(m,n,k)$. \\
	Let $c_r = 2^n (n^k2^k)^m$ be the size of the support of $D_R(m,n,k)$. \\
    Let $c_s$ be the size of the support of $D_S(m, n, k)$, or equivalently the number of satisfiable formulas in the support of $D_R(m, n, k)$.
	
	Every $\phi$ in the support of $D_R(m,n,k)$ appears in $f_{count}(\phi)$ tuples in the support of $\DPa(m,n,k)$.
	
	Thus, if we took the weighted sum of each tuple $(\vec{a},\phi)$ in $\DPa(m,n,k)$, giving each tuple containing a formula $\phi$ a weight of $1/f_{count}(\phi)$, this sum would equal $c_s$. In other words,
    
    $$\sum_{(\phi, \vec{a}) \sim \DPa} \frac{1}{f_{count}(\phi)} = c_s.$$
    
    Now we want to know what the probability is that $\phi$ drawn from $D_R(m,n,k)$ has at least one satisfying assignment. Call this probability $p_{sat} = c_s/c_r$.
    
	By definition, 
    
    $$y = E[f_{count}(\phi)]_{\phi \sim \DPa(m,n,k)} =  \sum_{(\vec{a},\phi)\in \DPa(m,n,k)}f_{count}(\phi)/c_p.$$
    
	Therefore, $\sum_{(\vec{a},\phi)\in \DPa(m,n,k)}f_{count}(\phi) = y \cdot c_p$.
	
	We want to minimize $\sum_{i=0}^{c_p} 1/x^i$ when $\sum_{i=0}^{c_p} x^i \leq y \cdot c_p$. These $x_i$ are implicitly representing values of $f_{count}(\phi)$. The sum is minimized when, for all $i$, we have $x_i=y$. This gives us $\sum_{i=0}^{c_p} 1/x^i \geq c_p/y$. So the minimum number of $\phi$ with a satisfying assignment is $c_s \geq c_p/y$.
	
	 Note that $p_{sat} = \frac{c_s}{c_r} \geq c_p/y \cdot 1/c_r$.
	
	Therefore, by Lemma \ref{lem:plantedToSatRealSetM}, $p_{sat} \geq e^{-cn/2^k}/y$.
\end{proof}

\begin{lemma}
	
Let $f_{count}(\phi)$ be the number of satisfying assignments of $\phi$.

Then, 
	$$E[f_{count}(\phi)]_{\phi \sim \DPa(m,n,k)} = \sum_{i=0}^n \binom{n}{i} \left(1- \frac{1-(1-(i/n))^k}{2^k-1}\right)^m.$$
	
\label{lem:expectedSatPlantedPrelim}
\end{lemma}
\begin{proof}
For a given $\vec{a}\in U_{\vec{a}}^n$ consider a random choice of vector $\vec{v}\in U_{\vec{a}}^n$ that has Hamming distance $i$ from $\vec{a}$. 

The probability that a random clause from $\phi$ is falsified by $\vec{v}$ is $\frac{1-(1-(i/n))^k}{2^k-1}$ by Corollary \ref{cor:boundSmallHDCl}. 
This implies that the probability that none of the $m$ randomly selected clauses are falsified is $\left(1-\frac{1-(1-(i/n))^k}{2^k-1}\right)^m$.

There are $\binom{n}{i}$ vectors with Hamming distance $i$. 

This gives $\binom{n}{i}(1-\frac{1-(1-(i/n))^k}{2^k-1})^m$ satisfied assignments in expectation at Hamming distance $i$. 

So the expected number of satisfying assignments is 
	$$E[f_{count}(\phi)]_{\phi \sim \DPa(m,n,k)} = \sum_{i=0}^n \binom{n}{i} \left(1- \frac{1-(1-(i/n))^k}{2^k-1}\right)^m.$$

\end{proof}

\begin{lemma}
	Let $m = (\ln(2)2^k-c)n$ for some constant $c$.
	
	Let $k\geq \kstr$.
	
	Let $f_{count}(\phi)$ be the number of satisfying assignments of $\phi$.
	
	Then, 
	$$E[f_{count}(\phi)]_{\phi \sim \DPa(m,n,k)} \leq n 2^{e^{-k/e^2} n} e^{2cn/2^k}.$$
	
	\label{lem:expectedSatPlanted}
\end{lemma}
\begin{proof}
Note that $k \geq \kstr \geq \kval$.	
	
By Lemma \ref{lem:expectedSatPlantedPrelim} we know that $E[f_{count}(\phi)]_{\phi \sim \DPa(m,n,k)} \le \sum_{i=0}^n \binom{n}{i} \left(1- \frac{1-(1-(i/n))^k}{2^k-1}\right)^m.$ 

Therefore, we want to show that $\sum_{i=0}^n \binom{n}{i} \left(1- \frac{1-(1-(i/n))^k}{2^k-1}\right)^m \le n 2^{e^{-k/e^2} n} e^{2cn/2^k}$ to prove the lemma's statement.

Let $g(i,n)$ be defined as:

$$g(i,n) =\binom{n}{i} \left(1- \frac{1-(1-(i/n))^k}{2^k-1}\right)^m.$$

We have that:

$$\sum_{i=0}^n g(i,n) \leq n\max_i g(i, n). $$

In the case work that follows, we will show that: 

$$\max_i g(i, n) \le 2^{e^{-k/e^2} n} e^{2cn/2^k}$$

Let $i = \eps n$ with $\eps \in [0,1]$, and let $p_f(i) = \frac{1-(1-(i/n))^k}{2^k-1}$.


We will now do some case work.

\begin{itemize}
	\item For $i\in [0, \frac{n}{k^2e^{k/e^2}}]$ 
	$$g(i,n) \leq \binom{n}{i} \left(1- \frac{1-(1-(i/n))^k}{2^k-1}\right)^m \leq \binom{n}{i} \leq \left(\frac{ne}{i}\right)^i \leq 2^{(\lg e + 2\lg k + \frac{k}{e^2}\lg e)\frac{n}{k^2 e^{k/e^2}}}\leq 2^{e^{-k/e^2} n}.$$
    
    The last inequality follows from the fact that $k \ge \kval$.
	
	
\item For $i \in [\frac{n}{k^2e^{k/e^2}}, \frac{n}{2k}]$ we will use three facts:

\begin{enumerate}
	\item $\binom{n}{i} \le (\frac{ne}{i})^i$
	\item $\left(1-p_f(i)\right)^m \le e^{-p_f(i) m}$
	\item $(1-(1 - (i/n))^k) \ge \frac{ik}{2n}$ when $i \in [ \frac{n}{k^2e^{k/e^2}} , \frac{n}{2k}]$
\end{enumerate} 
This yields the following bound:
$$g(i,n) \le \left(\frac{ne}{i}\right)^i e^{-p_f(i) m}$$

By the third fact, we have the following bound on $p_f(i)$:

$$p_f(i) = \frac{1-(1-(i/n))^k}{2^k-1} \ge \frac{ik}{2n(2^k-1)}$$

This yields a final bound of:

$$g(i,n) \le \left(\frac{ne}{i}\right)^i e^{-\frac{ik}{2n \cdot 2^k} m}$$

Plugging in our value for $m$:

$$g(i,n) \le \left(\frac{ne}{i}\right)^i e^{-\frac{ik}{2n \cdot 2^k}(2^k \ln(2) - c)n}$$   

Which simplifies to:
$$g(i,n) \le \left(\frac{ne}{i}\right)^i e^{-ik \ln(2)/2} e^{cik/(2 \cdot 2^k)}$$   

Using the fact that $ik/2 \le n$, we can rewrite this as:

$$g(i,n) \le \exp \left( \left( \ln \left(\frac{n}{i} \right ) +1 \right) i - \frac{ik \ln(2)}{2} \right) e^{cn/2^k}$$  

When $i \in [ \frac{n}{k^2e^{k/e^2}} , \frac{n}{2k}]$ we have that
$$\left( \ln \left(\frac{n}{i} \right ) +1 \right)-\frac{k \ln(2)}{2}<0 $$
for  $k\geq \kval$. 

Thus, for $i \in [\frac{n}{k^2e^{k/e^2}} , \frac{n}{2k}]$
$$g(i,n) \le e^{cn/2^k}.$$  

\item For $i \in [\frac{n}{2k}, \frac{n}{k}]$ we will use three facts:
\begin{enumerate}
	\item $\binom{n}{i} \le (\frac{ne}{i})^i$
	\item $\left(1-p_f(i)\right)^m \le e^{-p_f(i) m}$
	\item $(1 - (i/n))^k \le e^{-ik/n}$
\end{enumerate} 
We will use the same simplification based on the first two facts as above, but also using the new third fact this time:
$$g(i,n) \le \left(\frac{ne}{i}\right)^i e^{-(1-e^{-ik/n})n \ln(2)} e^{cn/2^k}.$$

Now note that over $i \in [\frac{n}{2k}, \frac{n}{k}]$, $(\frac{ne}{i})^i$ is monotonically increasing with $i$ and $e^{-(1-e^{-ik/n})n \ln(2)}$ is monotonically decreasing with $i$. Therefore, $(\frac{ne}{i})^i$ is maximized at $i = n/k$ and $e^{-\frac{1-e^{-ik/n}}{2^k}}$ is maximized at $i = n/(2k)$. We will bound the maximum of $g(i,n)$ by simultaneously maximizing both functions. 

Given that  $i \in [\frac{n}{2k}, \frac{n}{k}]$
$$g(i,n) \le \left(\frac{ne}{i}\right)^i e^{-(1-e^{-1/2})n \ln(2)} e^{cn/2^k}.$$ 

We can re-write this as:
$$g(i,n) \le \exp \left(  \left(\ln(n/i) +1\right) i -(1-e^{-1/2})n \ln(2) \right) e^{cn/2^k}.$$ 

$$g(i,n) \le \exp \left(  \frac{\ln(k) +1}{k} n -(1-e^{-1/2})n \ln(2) \right) e^{cn/2^k}.$$ 

When $k\geq \kval$,
$$\frac{\ln(k) +1}{k} n -(1-e^{-1/2})n \ln(2) <0.$$

Thus,
$$g(i,n) \le e^{cn/2^k}.$$ 

\item For $i \in [\frac{n}{k}, n/e^2]$ we will use three facts:
	\begin{enumerate}
		\item $\binom{n}{i} \le (\frac{ne}{i})^i$
		\item $\left(1-p_f(i)\right)^m \le e^{-p_f(i) m}$
		\item $(1 - (i/n))^k \le e^{-ik/n}$
	\end{enumerate} 
We will use the same simplification as above:
$$g(i,n) \le \left(\frac{ne}{i}\right)^i e^{-(1-e^{-ik/n})n \ln(2)} e^{cn/2^k}$$

Once again, note that over $i \in [\frac{n}{k}, n/e^2]$, $(\frac{ne}{i})^i$ is monotonically increasing with $i$ and $e^{-(1-e^{-ik/n})n \ln(2)}$ is monotonically decreasing with $i$. Therefore, $(\frac{ne}{i})^i$ is maximized at $i = n/e^2$ and $e^{-\frac{1-e^{-ik/n}}{2^k}}$ is maximized at $i = n/k$. We will bound the maximum of $g(i,n)$ by simultaneously maximizing both functions. 
	
$$g(i,n) \le (e^3)^{n/e^2} e^{-(1-e^{-1})n\ln(2)}e^{cn/2^k} \le \exp \left(\left(\frac{3}{e^2}+(\frac{1}{e}-1)\ln(2)\right)n \right) e^{cn/2^k}$$ 

This simplifies to:

$$g(i,n) \le \exp \left({\left(-0.03\right)n} \right)e^{cn/2^k} \leq 1$$
	

\item For $i\in[n/e^2,n/2]$ we will use the fact that $\binom{n}{i} \le 2^n$. It follows that:
	
    $$g(i,n) \leq 2^n 2^{-(1-(1-i/n)^k)n}e^{2cn/2^k} \le 2^{(1-i/n)^k n}e^{2cn/2^k} \leq 2^{e^{-k/e^2} n}e^{2cn/2^k}$$
    
	\item For $i\geq [n/2,n]$ we use the following two facts:
    \begin{enumerate}
        \item $\binom{n}{i} \le 2^n$
        \item $p_f(i) \ge 1/2^k$
    \end{enumerate}
	$$g(i,n) \leq \binom{n}{i} \left(1- p_f(i)\right)^m \leq 2^n e^{-m/2^k} \leq e^{2cn/2^k}.$$
\end{itemize}

This covers all the cases and gives us the desired result. 
\end{proof}

\begin{reminder}{Corollary \ref{cor:numSatIsHigh}}
Let $m = (\ln(2)2^k-c)n$.	
	
The probability that $\phi$ drawn from $D_R(m,n,k)$  has at least one satisfying assignment is at least $$\frac{1}{n}2^{- n/e^{k/e^2}}e^{-3cn/2^k}.$$
\label{cor:numSatIsHighRemind}
\end{reminder}
\begin{proof}
Use Lemma \ref{lem:plantedToRealSatCount} and set $y$ equal to the value in Lemma \ref{lem:expectedSatPlanted} to get this result. 
\end{proof}

\section{Putting it All Together}
\label{sec:putItTogether}
Recall that to prove our algorithm is correct and runs in $O\left(2^{n (1- \Omega(\lg^2(k)/k)}\right)$ time, we must bound the false positive and true positive rates. 
Now that we have given an upper bound on the false positive rate with Theorem \ref{thm:runningTimeIsGood} and given a lower bound on the true positive rate with Lemma \ref{lem:algIsCorrect}, we can prove the main theorem. 

\begin{reminder}{Theorem \ref{thm:HammDistCorrectFast}}
Assume $\phi$ is drawn from $D_{\Phi}(n,k)$. Let $\alpha = \frac{\lfloor n\lg(k)/(16k) \rfloor}{n}$.

Conditioned on there being at least one satisfying assignment to $\phi$, \HammDist($\phi$) will return some satisfying assignment with probability at least $ 1-3\cdot 2^{-n/(3\ln(2)2^k)}$.

Conditioned on there being no satisfying assignment to $\phi$, \HammDist($\phi$) will return False with probability $1$. 

\HammDist($\phi$) will run 
in time 
$$O\left(2^{n (1- \Omega(\lg^2(k)/k)}\right).$$

	\label{thm:HammDistCorrectFastRemind}
\end{reminder}
\begin{proof}
For $k<\kstr$, \LocalSearch~has a runtime of $2^{n(1-\gamma)}$ for some constant $\gamma>0$ \cite{localSearchAlg}. A constant is $\omega \left( \frac{\lg(k)}{k} \right)$, thus achieving our desired running time. \LocalSearch~succeeds with probability $1$ \cite{localSearchAlg}.
    
If there is a satisfying assignment and $k \geq \kstr$, then by Lemma \ref{lem:algIsCorrect}, there will be enough true positives that we will find one with high probability, and by Theorem \ref{thm:runningTimeIsGood}, there won't be so many false positives that we stop early and return False with high probability (as in lines 8-9 of Algorithm~\ref{alg:randomSatAlgo}). Using the bounds from Lema~\ref{lem:algIsCorrect} and Theorem~\ref{thm:runningTimeIsGood}, we can see that the algorithm is correct with probability greater than
$$1-e^{-n/(3\ln(2)2^k)}-3 e^{-(1-\alpha)^{2k}\ln(2)n/(32)}-2\cdot 2^{-n/(3\ln(2)2^k)},$$
which is greater than
$$1-3\cdot 2^{-n/(3\ln(2)2^k)}$$    

for large $n$ and $k \ge k^*$. If there is no satisfying assignment, then the algorithm will always be correct. The algorithm only returns SAT if some assignment is found. 
    

We now compute the running time. 
The running time of the algorithm is 
$$O\left(n^2 \cdot 2^n/\binom{n}{\alpha n} + |S|k^{\alpha n}\right).$$
We have that $|S|\leq 4 n^3 2^n/\left(\binom{n}{\alpha n} k^{\alpha n} \right) +1$, so we can simplify this further:
$$O\left(2^n/\binom{n}{\alpha n}\right).$$
Finally, we can plug in $\alpha = (1/16)\lg(k)/k$. Note that:

$$2^n / \binom{n}{\alpha n} \le 2^n \left(\frac{1}{\alpha}\right)^{-\alpha n}.$$

Thus we may conclude that our runtime is $$O\left(2^{n (1-\lg(1/\alpha)\alpha)}\right).$$
Now see that 
\begin{align}
\lg(1/\alpha)\alpha &= \frac{1}{20}(\lg(k)-\lg(\lg(k))+\lg(20))\lg(k)/k\\
&=\Omega(\lg^2(k))/k\\
\end{align}

for $k > \kstr$.

This gives us a final runtime of $O\left(2^{n (1-\Omega(\lg^2(k)/k))}\right)$.

\end{proof}

\section{Conclusion and Future Work}
\label{sec:conclusion}

We have presented a novel method for solving Random SAT, yielding a runtime faster than that of the previous best work~\cite{vyas2018super, PPSZ05} in the random case by a factor of $2^{n \Omega(\lg^2(k))/k}$. However, we conjecture that substantial improvements to the runtime of random-case SAT are still possible. Bellow we list what we feel are the most promising directions of future work. 

To begin with, we do not expect that our algorithm is the fastest of its kind, up to constants in the exponent (or, plausibly, asymptotic improvements in the exponent). We expect that simply by improving the test used when deciding whether or not to perform an expensive local search in the neighborhood of a randomly-sampled assignment, one can improve the performance of the algorithm. Additionally, on several occasions we use bounds that are not the tightest possible for simplicity's sake, and by tightening these bounds one can improve the constants in the exponent of our running time, and thus the asymptotic runtime of the algorithm for a large but fixed $k$. 


When extending this work and analyzing other similar algorithms for solving SAT, we note that Lemma \ref{lem:generalSmallHDf}~is a useful and general analysis tool. In particular, other sample-and-search algorithms that use a different test from ours may find it helpful to re-use the result of Lemma~\ref{lem:generalSmallHDf} for a different function $f$.


Our algorithm runs efficiently with high probability on any formula that has a sufficiently high value of $p_{TP}$ and a sufficiently low value of $p_{FP}$. We show that formulas drawn uniformly at random from those at the threshold have a high value of $p_{TP}$ and a low value of $p_{FP}$. As a result, given that $p_{TP}$ is large and $p_{FP}$ is small, our algorithm will run efficiently. This opens the door to improvement in the worst case if faster algorithms can be found for worst-case formulas where either $p_{TP}$ is too high or $p_{FP}$ is too low. 

The algorithm \HammDist~is non-constructive as currently written, because the value of the constant $\kstr$ is not currently known. Some constant $\kstr$ exists, because $\eps_k$ tends to zero with increasing $k$---indeed, $\eps_k = \tilde{O}(2^{-k/2})$, as shown by Coja-Oghlan and Panagiotou~\cite{boundingSATDensity2}, so we conjecture that this $\kstr$ will not need to be enormous before the exponential decay in $k$ causes $\eps_k$ to be exponentially small. This algorithm could be made constructive if explicit bounds with known constants were formulated for $\eps_k$, and thus $\kstr$.

Another avenue for improvement is analyzing the \HammDist~algorithm in the regime of small $k$. Our proofs rely on $k \geq \kstr \geq \kval$. However, we note that in our simulations, there was a noticeable improvement in the observed speed of the algorithm of Dantsin et al. when we introduced our test before searching in the neighborhood of an assignment. This improvement in speed was noticeable even for small $k$. Indeed, by looking at Fig.~\ref{fig:binom}, which was generated with real data for $k=4$, it is immediately apparent that there is a big difference between the distribution over how many clauses are satisfied by  small-Hamming-distance assignments as compared to random assignments. It seems plausible, even likely, that \HammDist~offers an improved running time even for small constant $k$ (e.g. $k=3,4,5$).

Finally, it is perhaps worth sparing a few words on the performance of our algorithm on real-world examples. In the previously mentioned simulations there is a noticeable difference between the distribution over how many clauses are satisfied by  small-Hamming-distance assignments as compared to random assignments. When searching for satisfying assignments via any variant of local search, it stands to reason that using a simple and cheap test before performing an expensive search in the neighborhood of a randomly-sampled assignment would yield improvements in small, practical problem instances, not only impossibly large ones. The reader who is interested in making the most practical version of our algorithm will likely find it useful to consider tests beyond the simple one we used, which will very likely improve the algorithm's practical performance still further.


\section{Acknowledgments}
\label{sec:ack}
We are very grateful to Greg Valiant, Virginia Williams, and Nikhil Vyas for their helpful conversations and kind support. Additionally, we are very appreciative of the email correspondence we had with Amin Coja-Oghlan, Alan Sly, and Nike Sun, who answered our many questions. We would also like to thank reviewers for their comments and suggestions. 


\let\realbibitem=\bibitem
\def\bibitem{\par \vspace{-0.5ex}\realbibitem}

\bibliographystyle{alpha}
\bibliography{mybib} 

\appendix

\section{Discussion}
\label{sec:Discussion}

Below we have three topics of discussion that did not fit in the main body of the paper. Subsection \ref{subsec:regularizeVStest} is a discussion of different ways to view the algorithmic framework of test and search. Readers who dislike the false positive and false negative framing may find this alternative analysis more appealing. Subsection \ref{subsec:motivation} further motivates our results. Subsection  \ref{subsec:kalphaVSexaustive} is primarily directed at readers who are interested in extending this work. 

\subsection{Alternate View on the Algorithmic Framework}
\label{subsec:regularizeVStest}

In our paper, we bound the false positive and true positive rates to show that our running time and correctness conditions are met. 

However, given our analysis and algorithm, one could take a different approach to proving the soundness of our algorithm. Let $T$ be the threshold we use in \HammDist~for the number of clauses an assignment must satisfy to justify a local search in its neighborhood. This alternate approach focuses on individual satisfying assignments and their small-Hamming-distance neighborhoods.

\begin{definition}
Call an assignment that satisfies a number of clauses above the  threshold, $T$, \emph{promising}. 

Let the set of assignments within Hamming distance $\alpha$ of a given assignment $\vec{a}$ be $H(\vec{a}, \alpha, n)$. 

Call an assignment, $\vec{a}$, a \emph{standard satisfying assignment} of a formula $\phi$ if both: (1) $\vec{a}$ satisfies the formula $\phi$ and (2) at least half of the assignments in $H(\vec{a}, \alpha, n)$ are \emph{promising}. 

Call a formula \emph{stuffed} if the total number of \emph{promising} assignments is more than $4n^3 \cdot 2^n/\left( \binom{n}{\alpha n} k^{\alpha n}\right)$.

Call a formula \emph{hollow} if the total number of \emph{promising} assignments at most $4n^3 \cdot 2^n/\left( \binom{n}{\alpha n} k^{\alpha n}\right)$.
\end{definition}

We know that $|H(\vec{a}, \alpha, n)| \geq \binom{n}{\alpha n}$. So, for any given standard satisfying assignment of the formula $\phi$, there are at least $\frac{1}{2} \cdot \binom{n}{\alpha n}$ assignments that are both promising and close enough in Hamming distance that when the local search is run, $\vec{a}$ will be found. 
	
It follows that for any given standard satisfying assignment, $\vec{s}$, if we sample $n^2 \cdot 2^n/\binom{n}{\alpha n}$ assignments and run local searches on assignments that are promising, we will find $\vec{s}$ with high probability. 

Additionally, if the formula is hollow, then \HammDist~runs in $O\left( 2^{n (1- \Omega(\lg(k))/k)} \right)$ time.

Now we give a proof sketch of our algorithm's correctness in terms of the new framework. 
We want to prove that formulas $\phi$ are hollow with high probability when drawn from $D_{\Phi}(n,k)$. We additionally want to prove that, conditioned on a formula being satisfiable, there exists a standard satisfying assignment of the formula $\phi$ with high probability. 

We will note that Lemma~\ref{lem:boundNumSamplesOurAlpha} does in fact prove that formulas drawn from $D_{\Phi}(n,k)$ are with high probability hollow. Furthermore, Corollary~\ref{cor:helpfulPlantedTo} proves that satisfied formulas drawn from $D_{\Phi}(n,k)$ have at least one standard satisfying assignment with high probability. 


We present a more general framework in the main body of our paper, but some readers may prefer the framework presented here.

\subsection{Additional Motivation}
\label{subsec:motivation}

We give some additional motivation for why our results are interesting, in particular our improved running time of $2^{n(1-\Omega(\lg^2(k))/k)}$. In `Tighter Hard Instances for PPSZ,' a distribution of hard examples for the PPSZ algorithm is produced \cite{harderExamplePPSZ}. On these instances PPSZ runs in time $2^{n(1-O(\lg(k)/k))}$. Given that both this lower bound on PPSZ and the algorithms of Vyas \cite{vyas2018super}~and Valiant \cite{Val} all run in $2^{n(1-O(\lg(k)/k))}$ time, one might believe that the best possible running time for random satisfiability is $2^{n(1-O(\lg(k)/k))}$ time. However, our algorithm, \HammDist, runs in time $2^{n(1-\Omega(\lg^2(k))/k)}$, which we consider remarkable. 




\paragraph{Number of Satisfying Assignments at the Threshold}

\newcommand{\numsol}{\textsc{NumSatSolutions}}

One question a skeptical reader might ask is: {\centering \emph{``How many satisfying assignments should I expect for satisfiable formulas at the threshold?''}} After all, if satisfiable formulas at the threshold have a very large number of satisfying assignments, then the problem of random satisfiability is easily solved. However, it is easy to show that the fraction of satisfying assignments at the threshold across all formulas is an exponentially small fraction of assignments. Specifically, the expected number of satisfying assignments in a randomly chosen formula, even conditioned on that formula being satisfied, is $O\left(2^{2n/e^{k/e^2}}\right)$. Note that $e^{k/e^2}$ grows exponentially faster than $\frac{k}{\lg^2(k)}$. 


So, although there may be many satisfying assignments at the threshold in an absolute sense, there are very few of them relative to the total number of assignments.

For notational convenience, let $\numsol(\phi)$ return the number of satisfying assignments of the formula $\phi$. 

\begin{lemma}
Let $m = d_k n = (\ln(2)2^k-c)n$.

Then,	
$$E_{\phi \sim D_{\phi}(n,k)}[\numsol(\phi)] = e^{cn/2^k}$$
and 
$$E_{\phi \sim D_{S}(d_kn,n,k)}[\numsol(\phi)] = O\left(2^{2n/e^{k/e^2}}\right).$$
\end{lemma}
\begin{proof}

First we will use Lemma \ref{lem:bassic numbers}.

The first notion of average number of satisfying assignments is the average over all formulas in $D_{\phi}(n,k)$:
$$E_{\phi \sim D_{\phi}(n,k)}[\numsol(\phi)] = e^{cn/2^k}.$$

However, we might be interested in the average number of satisfying assignments a formula has conditioned on it having at least one satisfying assignment. Let $p_{SAT}(n,k)$ be the probability that a formula drawn from $D_{\phi}(n,k)$ is satisfied. Recall also that $D_{S}(d_kn,n,k)$ is the distribution of $D_{\phi}(n,k)$ conditioned on the formula being satisfied. Therefore,
$$E_{\phi \sim D_{S}(d_kn,n,k)}[\numsol(\phi)] = e^{cn/2^k} \frac{1}{p_{SAT}(n,k)}.$$

Now we can use Lemma \ref{cor:numSatIsHighRemind}. We have that $\frac{1}{p_{SAT}(n,k)} \leq n 2^{ n/e^{k/e^2}}e^{3cn/2^k}$. So,
$$E_{\phi \sim D_{S}(d_kn,n,k)}[\numsol(\phi)] \leq  n 2^{ n/e^{k/e^2}}e^{4cn/2^k}  = O\left(2^{2n/e^{k/e^2}}\right).$$
\end{proof}

\subsection{Exhaustive Search}
\label{subsec:kalphaVSexaustive}

\newcommand{\shds}{SHDS}

In Section \ref{sec:algorithm} we discuss the Small Hamming Distance Search Algorithm (\shds) \cite{localSearchAlg}. Given an assignment at Hamming distance at most $\alpha n$ from a satisfying assignment, \shds~will return a satisfying assignment in time $k^{\alpha n}$. 

An alternative exhaustive search algorithm requires $O^*(\binom{n}{\alpha n})$ time to find a satisfying assignment given an assignment that is at Hamming distance at most $\alpha n$ from a satisfying assignment. When $\alpha \leq 1/3$, trying all assignments at distance at most $\alpha n$ takes time:
$$\sum_{i = 0}^{\alpha n} \binom{n}{i} m = m \cdot \frac{1-\alpha}{1-2\alpha} \binom{n}{\alpha n} = O^* \left( \binom{n}{\alpha n} \right).$$

Note that when $\alpha = \omega \left( 1/k \right)$,
$$k^{\alpha n} \geq \left(\frac{e}{\alpha}\right)^{\alpha n}  \geq \binom{n}{\alpha n}.$$
This improvement doesn't result in an improvement over the running time of  $2^{n (1- \Omega(\lg^2(k))/k)}$ presented in this paper. However, as $\alpha$ gets smaller, the improvement becomes more important. Notably, if there were a perfect test, one for which $p_{FP}=0$, then $\binom{n}{\alpha n}$ significantly outperforms $k^{\alpha n}$. Specifically, if $p_{FP}=0$, then the test and search framework produces an algorithm with running time $2^{n (1- 1/O(\lg(k)))}$ when using the $k^{\alpha n}$ search, but an algorithm with running time $O^*(2^{n/2})$ algorithm when using the $\binom{n}{\alpha n}$ search. 



	
	
	
	



\end{document}